\documentclass[12pt]{article}
\usepackage{amsfonts}
\usepackage{amsmath, amsthm, amssymb, bbm}
\usepackage{xcolor,hyperref,multirow,graphicx}
\usepackage{float}
\usepackage[title]{appendix}
\usepackage{natbib}
\usepackage{subcaption}
\usepackage{enumerate}
\usepackage[inline, shortlabels]{enumitem}
\usepackage[multiple]{footmisc}
\usepackage[font=small,labelfont=bf,tableposition=top]{caption}

\DeclareCaptionLabelFormat{andtable}{#1~#2  \&  \tablename~\thetable}

\providecommand{\keywords}[1]
{
  \small	
  \textbf{\textit{Keywords---}} #1
}

\newtheorem{assumption}{Assumption}

\newtheorem{theorem}{Theorem}

\newtheorem{lemma}{Lemma}
\newtheorem{corollary}{Corollary}
\newtheorem{proposition}{Proposition}

\newcommand{\plim}{\operatorname*{plim}}

\newcommand{\argmin}{\operatorname*{argmin}}

\newcommand{\norm}[1]{\left\lVert#1\right\rVert}

\newcommand\independent{\protect\mathpalette{\protect\independenT}{\perp}}
\def\independenT#1#2{\mathrel{\rlap{$#1#2$}\mkern2mu{#1#2}}}

\usepackage{mathtools}

\theoremstyle{definition}
\newtheorem{algo}{Algorithm}

\begin{document}

\title{\Large Linear Multidimensional Regression\\ with 
Interactive Fixed-Effects
}
 \author{
  \large Hugo Freeman\footnote{
                    freem391@msu.edu,
                   Department of Economics,
                   Michigan State University,
                   East Lansing,
                   USA.
                   }
}
\date{\today}
\maketitle
\vspace{-0.8cm}
\begin{abstract}
    This paper studies a linear model for multidimensional panel data of three or more dimensions with unobserved interactive fixed-effects. The main estimator is a weighted-within transformation, which is Neyman orthogonal and requires one preliminary step. In that step, the model is embedded within a two-dimensional panel framework where factor model methods in \cite{Bai2009} lead to consistent, but slowly converging, estimates that serve as proxies for the fixed-effects. The weighted-within transformation is robust to multidimensional interactive fixed-effects and achieves the parametric rate of consistency. The estimator is shown to be asymptotically normal. The methods are implemented to estimate the demand elasticity for beer. 
\end{abstract}

\keywords{Multidimensional panel data; 
    interactive fixed-effects}

\section{Introduction}

Models of multidimensional data -- data with more than two dimensions -- are becoming increasingly important in econometric analysis as large data sets with a multidimensional structure become available. 
For example, consider studying demand elasticities with consumption data that varies by product, $i$, store, $j$, repeated over time, $t$.\footnote{A non-exhaustive list of related examples can be found in the introduction of \cite{matyas2017econometrics}. } 
In this context, analysts may be concerned with shifts in taste preferences unobserved by the econometrician related to unobserved characteristics in each dimension - like a cultural or sporting event that impacts prices and demand heterogeneously over product, store, and time.  
Whilst these are not explicitly observed by the econometrician, fixed-effects are a useful tool to  control for such variation. 

Additive fixed-effects in multidimensional data can at most accommodate variation in unobserved heterogeneity over a subset of dimensions with any of the fixed-effects terms. 
For example, in the three-dimensional model, additive effects only control for variation over $ij$, $it$ and $jt$, but not jointly over all $ijt$. 
Hence, if heterogeneity is across all dimensions interactively, additive fixed-effects cannot sufficiently control for unobserved heterogeneity. 
This paper develops tools to control for unobserved heterogeneity in the form of interactive fixed-effects that controls for variation that interacts over all dimensions of the data.

Consider $\beta$ estimation in the interactive fixed-effects model with three dimensions.
The estimator considered herein controls for interactive fixed-effects of the form,\footnote{This decomposition is considered in \cite{babii2025tensor} for the pure factor model (no covariates). } 
\begin{align}\label{eqn:IntroTuckerModel}
    Y_{ijt} = X_{ijt}^\prime\beta + 
    \sum_{r_1=1}^{R_1} \sum_{r_2=1}^{R_2} \sum_{r_3=1}^{R_3} 
    \lambda_{i, r_1}  \delta_{j, r_2}  \gamma_{t, r_3} +
    \varepsilon_{ijt}.
\end{align}
Whilst the estimator controls fixed-effects in \eqref{eqn:IntroTuckerModel}, the following version is motivated throughout the main paper for notational simplicity, and the general version in \eqref{eqn:IntroTuckerModel} is relegated to Appendix~\ref{sect:AppendixIterative}. 
Consider instead, 
with all terms in $\sum_{\ell = 1}^L \lambda_{i\ell}  \delta_{j\ell}  \gamma_{t\ell}$ are unobserved:\footnote{Bounded $L$ is only needed for the parametric rate of consistency, which requires the multilinear rank, defined in Section~\ref{sect:notation}, to be bounded in every dimension. Less sparse models with $L$ unbounded are possible, but lead to slower convergence. }
\begin{align}\label{eqn:IntroModel}
    Y_{ijt} = X_{ijt}^\prime\beta + 
    \sum_{\ell = 1}^L \lambda_{i\ell}  \delta_{j\ell}  \gamma_{t\ell} +
    \varepsilon_{ijt}. 
\end{align} 

Considering three dimensions is without loss of generality for the methods used herein. 
The object $L$ is difficult to calculate, see \cite{haastad1989tensor}, however, for purposes of this analysis it is only necessary to know the multilinear rank, defined in Section~\ref{sect:notation} as the $\{R_1,R_2,R_3\}$ from \eqref{eqn:IntroTuckerModel}. 
Estimation of the rank of interactive fixed-effect in the linear regression model is still an open area of research in the standard panel setting, however, simulations suggest estimation with more than the required components performs well.\footnote{In standard panel data overestimating rank leads to valid inference, see \cite{MoonWeidner2015}. It is conjectured this result translates to the multidimensional case. } 
For purposes of this paper, the rank $L$, respectively $\{R_1,R_2,R_3\}$ from \eqref{eqn:IntroTuckerModel}, is considered known. 
Additive fixed-effects are omitted for brevity but can be removed with a simple within transformation. 

Let $X_{ijt}$ be arbitrarily correlated with interactive fixed-effects but independent of the noise term, $\varepsilon_{ijt}$. 
The challenge to estimating $\beta$ is isolating variation in $X_{ijt}$ that is not correlated with the interactive fixed-effects term. 
This paper develops multidimensional weighted-within transformations that control this unobserved heterogeneity and shows settings where standard factor methods work well. 
The weighted within transformation is a novel contribution. 
Group fixed-effects from \cite{bonhomme2022discretizing}, and in \cite{freeman2023linear} are also possible, but are not smooth estimators, so may have less favourable properties when fixed-effects are continuous, see \cite{andrews1994asymptotics,freeman2026inference}.

This paper makes two main contributions to the literature. 
The first is to show that the three or higher dimensional model can be couched in a standard two-dimensional panel data model, and to derive sufficient conditions for consistency using factor model methods from \cite{Bai2009,MoonWeidner2015}, albeit at slow rates of convergence. 
The second contribution is to introduce weighted fixed-effects methods, which are Neyman orthogonal and achieve the parametric rate of convergence and asymptotic normality for $\beta$ estimates.
This is yet to be shown in the three dimensional case using existing methods. 

To understand the novel estimator proposed in this paper, it is useful to begin with the familiar within transformation. 
Consider additive fixed-effects of the form $a_{ij} + b_{it} + c_{jt}$. 
These can be projected out using the transformation,
\begin{align}\label{withinEstimator}
        \dot{{Y}}_{ijt} = {Y}_{ijt} - \Bar{{Y}}_{\cdot jt} - \Bar{{Y}}_{i\cdot t} - \Bar{{Y}}_{ij\cdot} + \Bar{{Y}}_{\cdot\cdot t} + \Bar{{Y}}_{\cdot j\cdot}+ \Bar{{Y}}_{i\cdot \cdot} - \Bar{{Y}}_{\cdot \cdot \cdot}, 
\end{align}
applied equivalently to $X_{ijt}$, where the bar variables denote the average taken over the ``dotted'' index for the entire sample. 
That is, 
$\Bar{{Y}}_{\cdot jt} := \frac{1}{N_1}\sum_{i = 1}^{N_1} Y_{ijt}$, 
$\Bar{{Y}}_{\cdot\cdot t} := \allowbreak \frac{1}{N_1N_2} \sum_{i = 1}^{N_1}\sum_{j = 1}^{N_2}\allowbreak Y_{ijt}$, etc. 
In the presence of interactive fixed-effects, the simple within transformation demeans each fixed-effect term to leave,
\begin{align*}
    \sum_{\ell = 1}^L 
    \Big(\lambda_{i\ell} - \Bar{\lambda}_{\ell}\Big)
    \Big(\delta_{j\ell}  - \Bar{\delta}_{\ell}\Big)
    \Big(\gamma_{t\ell}  - \Bar{\gamma}_{\ell}\Big)
\end{align*}
which is clearly not controlled for if the $\{\lambda,\delta,\gamma\}$ terms admit heterogeneity over $i,j,t$. 

Consider instead an extension of the within transformation that uses weighted means instead of uniform means. 
With an abuse of notation, consider,  
\begin{align}\label{withinGroupEstimator}
    \begin{split}
        \check{Y}_{ijt} &= Y_{ijt} - \Bar{Y}_{i^*jt} - \Bar{Y}_{ij^* t} - \Bar{Y}_{ijt^*} + \Bar{Y}_{i^*j^* t} + \Bar{Y}_{i^* jt^*}+ \Bar{Y}_{ij^* t^*} - \Bar{Y}_{i^*j^*t^*}
    \end{split}
\end{align}
where the bar variables combined with the star indices denote weighted means for that observation. 
For example, $\Bar{Y}_{i^*jt} = \sum_{i^\prime}w_{i,i^\prime}Y_{i^\prime jt}$ for weights across $i^\prime$ for each $i$. 
In turn, the remainder from the interactive fixed-effects term is,
\begin{align*}
    \sum_{\ell = 1}^L 
    \Big(\lambda_{i\ell} - \sum_{i^\prime}w_{i,i^\prime}{\lambda}_{i^\prime\ell}\Big)
    \Big(\delta_{j\ell} - \sum_{j^\prime}w_{j,j^\prime}{\delta}_{j^\prime\ell}\Big)
    \Big(\gamma_{t\ell} - \sum_{t^\prime}w_{t,t^\prime}{\gamma}_{t^\prime\ell}\Big).
\end{align*}
If appropriate weights are used, the weighted within transformation can project out the interactive fixed-effects. 
The same can be true if the weighted means are replaced with cluster means for appropriate cluster assignments, similar to a group fixed-effects estimator.
Hence, with a relatively small change to how the within transformation is performed, much more general fixed-effects can be considered in the linear model. 

The model for interactive fixed-effects has precedent in the standard two-dimensional panel data setting, for example in \cite{Bai2009,Pesaran2006},
\begin{align}\label{eqn:bai}
    Y_{it} = X_{it}^\prime\beta + \sum_{\ell = 1}^L\lambda_{i\ell} f_{t\ell} + e_{it}.
\end{align}
The interactive term $\sum_{\ell = 1}^L\lambda_{i\ell} f_{t\ell}$ also sufficiently captures variation in additive individual and time effects without the need to specify these separately. 
For multidimensional applications the problem in \eqref{eqn:IntroModel} can be transformed to a two dimensional problem and estimated as \eqref{eqn:bai} directly using the transformed data.
Indeed, Section~\ref{sect:Matrixestimation} displays useful preliminary convergence rates for this approach. However, 
convergence rates can suffer severely from the over-parametisation implied by \eqref{eqn:bai}. 
Without strong assumptions on the sparsity of the fixed-effects, only a slow rate of convergence can be guaranteed for this approach. 
However, when this approach is used to construct preliminary estimators of the fixed-effects, the parametric rate of consistency is possible when these are combined with the novel weighted within transformations. 
\cite{kuersteiner2020dynamic} consider a similar dimension specific projection in the time dimension in the presence of error correlations, in particular see Supplement D.5.2. 

There is a growing literature on interactive fixed-effects over three or more dimensions of data. 
Two examples that consider the linear regression model are \cite{kapetanios2021estimation, feng2024estimation}, which consider dense versions of the fixed-effects in \eqref{eqn:IntroModel}. 
That is, they consider interactions that are full rank over multiple indices, for example, $\lambda_{ij}'f_t + \tilde \lambda_{ij}' f_{jt}$. 
Appendix~\ref{sect:AppendixIterative} shows this paper's identification strategy and associated estimator can be specified to also capture fixed-effects of the form, e.g. $\lambda_{ij}'f_t + \tilde \lambda_{ij}' f_{jt}$, but will come at the cost of non-negligible asymptotic bias. 
This does not necessarily void consistency, but does imply slower rates of convergence. 
For example, rates of convergence in \cite{kapetanios2021estimation,feng2024estimation} reduce to $\sqrt{T}$; much slower than the potential parametric rate of $\sqrt{N_1N_2 T}$. 
Additionally, assuming fixed-effects enter $X$ in parametric form as per \cite{kapetanios2021estimation} is potentially restrictive. 
Nonparametric methods considered herein allow these to enter $X$ flexibly. 

The models in \cite{kapetanios2021estimation, feng2024estimation} also consider heterogeneous coefficients of the form $\beta_{ij}$. 
Indeed, this leads to a more flexible model, but as in the case with the dense fixed-effect term $\lambda_{ij}'f_t + \tilde \lambda_{ij}' f_{jt}$, this leads to a loss of precision in $\beta$ estimates. 
Convergence rates for $\beta_{ij}$ in this setting can be at most $\sqrt{T}$, orders of magnitude slower than $\sqrt{N_1N_2 T}$. 
By considering the sparse specification in \eqref{eqn:IntroTuckerModel} and \eqref{eqn:IntroModel}, this paper gives the analyst additional tools to trade-off flexibility and precision in their modelling preferences. 

Finite sample bias may also arise when rank of the fixed-effect term is large and only a subset of the unobserved heterogeneity parameters are low-dimensional. 
For unbiased estimates of $\beta$, transforming the multidimensional array to a matrix then estimating \eqref{eqn:bai} requires either: (a) all the fixed-effects are low-dimensional; or (b) 
that a known subset of the fixed effects are low-dimensional, which is a strong assumption.
Alternatively, whilst the weighted transformations also require that a subset of the fixed-effect parameters are low-dimensional, the analyst does not need to know which ones are.
In this sense, the novel weighted estimator is robust, see Corollary~\ref{cor:robust}.
A concrete example is considered in a simulation exercise, and there is evidence in the empirical application of differences in the rank of fixed-effects in each dimension.

The beer demand elasticity application uses Dominick's supermarket data for Chicago, 1991-1995, where price and quantity vary over product, store, and fortnight.
For comparison, barley price is used as an instrument to estimate elasticities.
IV estimates show demand is strongly negatively elastic (-3.39), however, estimates are very imprecise. 
This could be because the instrument only varies over one dimension, time, or, e.g., because prices of other inputs are also highly variable so the instrument does not explain much price variation. 
Estimates from the weighted within transformation also show demand is downward sloping and elastic (-2.99 to -4.11 for three to five interactive terms), but with much better precision. 
Factor model estimates are highly sensitive to how the data is transformed into two-dimensions. 
The estimates from the weighted-within transformation are more elastic than the own-price elasticities estimated in \cite{hausman1994competitive}. 

The technical component of this paper relates to the numerical analysis literature on low-rank approximations of multidimensional arrays. 
As pointed out in \cite{de2008tensor}, the low-rank approximation problem in the tensor setting is not well-posed, hence presents technical difficulties. 
See \cite{kolda2009tensor} for a summary of the multidimensional array decomposition problem, and \citet{vannieuwenhoven2012new, rabanser2017introduction}. 
Therefore, it is necessary to innovate on this tensor low-rank problem to find appropriate analytical results. 
To this end, this paper utilises well-posed components of  
two-dimensional sub-problems
for use in nuisance parameter applications. 
The advantage in this paper is that fixed-effects only need to be differenced out at a sufficiently fast asymptotic rate, and the low-rank tensor problem does not need to be solved explicitly.\footnote{\cite{elden2011perturbation}, along with related papers, suggest a reformulation of the low multilinear rank problem that may have promising applications. }

A related statistics literature estimates factor models for tensor time series, see \cite{chen2022factor,han2024tensor}, and \cite{babii2025tensor} in econometrics.
There the loadings and factors are the objects of interest, there are no regressors, and the low-rank structure is assumed in every dimension.
Here the interactive fixed-effects are a nuisance term correlated with the regressors, they only need to be differenced out fast enough for inference on $\beta$, and consistency only needs them to be low-rank along one flattening, which the analyst does not need to know.
The closest work is on semiparametric two-way panels.
\cite{bonhomme2022discretizing} and \cite{freeman2023linear} control two-way heterogeneity with group fixed-effects, and \cite{freeman2026inference} study the two-way version of the weighted-within transformation and its Neyman orthogonal moment condition.
Adding a third dimension changes the problem in two ways.
First, the singular value decomposition no longer gives a best low-rank approximation, so proxies must come from separate flattenings, any of which may fail to be low-rank.
Second, the proxies may then be accurate in some dimensions only, and because the weighted-within transformation differences out the fixed-effects along every dimension, the remainder is a product of proxy errors across dimensions, which delivers the robustness in Corollary~\ref{cor:robust}.

The paper is organised as follows:
Section~\ref{sect:motivation} introduces the general model;
Section~\ref{sect:estimators} details the estimators and convergence results; 
Section~\ref{sect:Discussion} discusses the convergence results and some practical considerations; 
Section~\ref{sect:Application} demonstrates the estimator with a beer demand estimation empirical application;
Section~\ref{sect:sims} displays the simulation results;
and
Section~\ref{sect:conclusion} concludes.

\section{Model}\label{sect:motivation}

Let $\beta^0$ denote the true parameter value for the slope coefficients. 
The model in full dimensional generality is,
\footnote{
For example, in index notation this model can be written as, 
\begin{align*}
    Y_{i_1, i_2, \dots, i_d} = \sum_{k = 1}^K X_{i_1, i_2, \dots, i_d; k}\beta_k^0 
    + \mathcal{A}_{i_1, i_2, \dots, i_d} 
    + \varepsilon_{i_1, i_2, \dots, i_d}
\end{align*}
with $\mathcal{A}_{i_1, i_2, \dots, i_d} = \sum_{\ell = 1}^L \varphi^{(1)}_{i_1\ell} \dots  \varphi^{(d)}_{i_d\ell}$.
}
\begin{align}\label{eqn:fullModel}
    \boldsymbol{Y} = \sum_{k = 1}^K \boldsymbol{X}_k \beta_k^0  
    + \boldsymbol{\mathcal{A}} 
    + \boldsymbol{\varepsilon}, 
\end{align}
where $\boldsymbol{Y}, \boldsymbol{X}_k, \boldsymbol{\varepsilon} \in\mathbb{R}^{N_1\times N_2\times\dots \times N_d}$.
$\boldsymbol{\mathcal{A}} = \sum_{\ell = 1}^L \varphi^{(1)}_\ell \circ \dots\circ \varphi^{(d)}_\ell$ where $\varphi^{(n)}_\ell\in\mathbb{R}^{N_n}$ for each $n = 1, \dots, d$ and ``$\circ$'' is the outer product. 
$L$ need not be bounded, since the results below restrict the multilinear rank of $\boldsymbol{\mathcal{A}}$, defined in Section~\ref{sect:notation}, rather than $L$.\footnote{Appendix~\ref{sect:AppendixIterative} shows these estimators control for variation from a Tucker decomposition. }
$\boldsymbol{\varepsilon}$ is a noise term independent of all $\boldsymbol{X}_k$ and all unobserved fixed-effects terms. 
Take $i_n\in\{1,\dots,N_n\}$ for all $n\in\{1,\dots,d\}$ as the dimension specific index, where $N_n$ is the sample size of dimension $n$. 
The regressors $\boldsymbol{X}_k$ may be arbitrarily correlated with $\boldsymbol{\mathcal{A}}$. 
Throughout this paper all dimensions are considered to grow asymptotically, that is $N_n \rightarrow \infty$ for all $n$, with $N_n \lesssim \prod_{n^\prime \neq n} N_{n^\prime}$.
This allows the dimensions to grow at very different rates, for example $N_1 \asymp N_2N_3$ when $d = 3$, and is milder than the rate $N \asymp T$ common in two-dimensional panels.

Model \eqref{eqn:fullModel} can be seen as a natural extension of the \cite{Bai2009} model to three (or more) dimensions with $\boldsymbol{\mathcal{A}}$ interpreted as a ``higher-dimensional'' factor stucture. 
Similar to this strand of the literature, $\boldsymbol{\mathcal{A}}$ is a nuisance term, and ``fixed'' means that its dependence with the regressors is not restricted to a parametric form, not that it is non-random.
The components $\varphi^{(n)}_{i_n}$ are random variables, and the conditional expectations and moment conditions below are taken with respect to their joint distribution with $\boldsymbol{X}_k$ and $\boldsymbol{\varepsilon}$.
There are potentially many extensions to the factor model setting in \cite{Bai2009} to the higher dimension case. 
This paper starts with what seems the most natural extension. 

$\boldsymbol{\mathcal{A}}$ incorporates additive fixed effects that vary in any strict subset of dimensions. 
For example, in the three dimensional setting it can control for the additive terms, $a_{ij} + b_{it} + c_{jt}$. 
These can be controlled for using $L = \min\{N_1, N_2\} + \min\{N_1, N_3\} + \min\{N_2, N_3\}$, 
with the first $\min\{N_1, N_2\}$ terms $\sum_{\ell = 1}^{\min\{N_1, N_2\}} \varphi_{i\ell}^{(1)}\varphi_{j\ell}^{(2)} = a_{ij}$ by making $\varphi_{t\ell}^{(3)}$ constant for $\ell = 1,\dots, \min\{N_1, N_2\}$, 
and so on for the $ b_{it}$ and $c_{jt}$.\footnote{Strictly speaking, two or more such terms give every flattening of $\boldsymbol{\mathcal{A}}$ a rank that grows with the sample sizes, which the rank conditions below do not allow. }
These could also be controlled for directly using the standard within-transformation before considering the model in \eqref{eqn:fullModel}.

\subsection{Notation and preliminaries}\label{sect:notation}

For a $d$-order tensor, $\mathbf{A}$, a factor-$n$ flattening, denoted as $\mathbf{A}_{(n)}$, is the rearrangement of the tensor into a matrix with dimension $n$ varying along the rows and the remaining dimensions simultaneously varying over the columns. 
That is, $\mathbf{A}_{(n)}\in\mathbb{R}^{N_n\times N_{n+1}N_{n+2}\dots N_1\dots N_{n-1}}$. 
The Frobenius norm, $\|\cdot\|_F$, of a matrix or tensor is the entry-wise norm, $\|\mathbf{A}\|_F^2 = \sum_{i_1 = 1}^{N_1} \dots \sum_{i_d = 1}^{N_d} A_{i_1\dots i_d}^2$. 
The spectral norm, denoted $\| \cdot\|$, is the largest singular value of a matrix. 
For a $d$-order tensor, $\mathbf{A}$, the multilinear rank, denoted $\mathbf{r}$, is a vector of matrix ranks after factor-$n$ flattening in each dimension, with each component of this vector $r_n = rank\big(\mathbf{A}_{(n)}\big)$. 
Tensor rank, different to multilinear rank, is defined as the least number of outer products of vectors to replicate the tensor. 
That is, for tensor $\mathbf{A}$ and vectors $u_\ell^{(n)}\in\mathbb{R}^{N_n}$, tensor rank is the smallest $L$ such that $\mathbf{A} = \sum_{\ell = 1}^L u_\ell^{(1)}\circ \dots\circ u_\ell^{(d)}$, where $\circ$ is the outer product of a vector. 
The notation $a\lesssim b$ means the asymptotic order of $a$ is bounded by the asymptotic order of $b$. 
The $n$-mode product of a tensor $\mathbf{A}$ and matrix $B$ is denoted $\mathbf{A}\times_n B$ and has elements 
$$(\mathbf{A}\times_n B)_{i_1,\dots,j,\dots,i_d} = \sum_{i_n = 1}^{N_n} A_{i_1,\dots, i_n, \dots, i_d} B_{j,i_n}, $$
which is equivalent to saying the flattening $(\mathbf{A}\times_n B)_{(n)} = B\mathbf{A}_{(n)}$.

The singular value decomposition of a matrix, ${A}\in \mathbb{R}^{N_1\times N_2}$ is 
\begin{align}\label{eqn:SVD}
    {A} = U\Sigma V^\prime = \sum_{r = 1}^{\min\{N_1,N_2\}} \sigma_r u_r v_r^\prime, 
\end{align}
where $U$ is the matrix of left singular vectors, $u_r$, $V$ is the matrix of right singular vectors, $v_r$, and $\Sigma$ is a diagonal matrix of singular values, $\sigma_r$, with values running in descending order down the diagonal. 
For a rank-$r$ matrix, the first $r$ entries on the diagonal of $\Sigma$ are strictly positive and the remaining entries are zero. 

Take the approximation problem, 
\begin{align}
    \min_{A^\prime}\|A - A^\prime\|_F \, \textrm{ such that }rank(A^\prime) \leq k.
\end{align}
It is well known from the Eckart-Young-Mirsky theorem that the solution to this approximation problem is the first $k$ terms of the singular value decomposition, i.e. $\sum_{r = 1}^{k} \sigma_r u_r v_r^\prime$. 
The Eckart-Young-Mirsky theorem effectively picks out the row and column subspaces that best explain variation in the matrix $A$ as the leading columns of the matrix $U$, respectively of $V$. 
The sum of squared error at the minimiser is thus $ \sum_{r = k + 1}^{\min\{N_1,N_2\}} \sigma_r^2$. 
This is commonly called a low-rank approximation and forms the cornerstone for estimation of unobserved heterogeneity in the factor model and interactive fixed-effects models in \citet{BaiNg2002,Bai2009, MoonWeidner2015} amongst others. 

The Eckart-Young-Mirsky theorem, however, does not extend to the three or higher dimensional setting, see \cite{de2008tensor} for details. 
To avoid this complication, the multidimensional problem can be translated to the two-dimensional setting to utilise the Eckart-Young-Mirsky theorem, or the fixed-effects parameters need to be shrunk separately, as is done with the weighted-within transformation.

\section{Estimation}\label{sect:estimators}

The weighted-within estimator in this section draws on the double debias approach summarised in \cite{chernozhukov2022locally}, and in this linear setting resembles closely \cite{robinson1988root}. 
Studying Neyman orthogonal properties has become useful in the literature, see e.g. \cite{bonhomme2024neyman}.

Section~\ref{sect:Matrixestimation} details a matrix method that provides proxies for the fixed-effects. 
Section~\ref{sect:estimatorsKernel} introduces the weighted-within estimator, which uses these proxies to estimate the nonparametric components of $\boldsymbol{Y}$ and $\boldsymbol{X}$ that depend on the fixed-effects, and differences them out along every dimension. 
Section~\ref{sect:inference} shows that the weighted-within estimator is Neyman orthogonal and establishes its asymptotic normality. 
The weighted-within estimator is novel to this paper. 
The matrix method in Section~\ref{sect:Matrixestimation} is also novel, but is a very simple adaptation of two-dimensional estimators, and may be of independent interest. 

\subsection{Matrix low-rank approximation estimator}\label{sect:Matrixestimation}

This section provides a description of some matrix methods that can be applied directly to the multidimensional model and stipulates the assumptions required for consistency. 
\cite{kapetanios2021estimation} employ a similar approach for three-dimensional arrays in conjunction with the \cite{Pesaran2006} common correlated effects estimator. 
\cite{babii2025tensor} employ a similar matricisation procedure as that detailed below, but are interested in inference on the fixed-effect parameters.

Consider recasting the multidimensional array problem into a two dimensional panel problem by flattening $Y$ and $X$ in the $n$-th dimension,
\begin{align*}
    Y_{(n)} = \sum_{k=1}^K X_{(n),k}  \beta_k ^0 + \varphi^{(n)} \Gamma^\prime_n +\varepsilon_{(n)}
\end{align*}
where $Y_{(n)}, X_{(n),k} ,\varepsilon_{(n)} \in \mathbb{R}^{N_n\times \prod_{n^\prime \neq n}N_{n^\prime}}$, $\varphi^{(n)}$ is an $N_n\times r_{n}$ matrix and $\Gamma_n$ is an $\prod_{n^\prime \neq n}N_{n^\prime}\times r_{n}$ matrix that accounts for variation in the remaining $\varphi^{(n^\prime)}$ for all $n^\prime\neq n$. 
The term $r_{n}$ is indexed by the dimension $n$ because it may vary non-trivially according to the flattened dimension. 
This is then the model described in \eqref{eqn:bai}, that is, the standard linear model with factor structure unobserved heterogeneity as studied in \cite{Bai2009}. 

The two-dimensional fixed-effects estimator for a given flattening, $n$, optimises
\begin{align}\label{eqn:factorOptimiser}
    R(\beta, \widehat{r}_n, n) = 
    \min_{\substack{\varphi^{(n)}\in \mathbb{R}^{N_n\times \widehat{r}_n},\\
    \Gamma_n\in \mathbb{R}^{\prod_{n^\prime \neq n}N_{n^\prime} \times \widehat{r}_n}}} \left\| Y_{(n)} - \sum_{k=1}^K X_{(n),k}  \beta_k - \varphi^{(n)} \Gamma_n^\prime \right\|_F^2.
\end{align}
Then $\widehat{\beta}^{2D}_{(n)} = \argmin_{\beta} R(\beta, \widehat{r}_n, n)$ is the slope estimate for the two-dimensional setup. 
The analyst must choose both the dimension to flatten in, $n$, and the rank of the estimated interactive fixed-effects term, $\widehat{r}_n$. 
It is well known that the minimum in \eqref{eqn:factorOptimiser} is achieved using the leading $\widehat{r}_n$ terms from the singular value decomposition of the error term, $Y_{(n)} - X_{(n)}^\prime \beta$.
This gives $\widehat{\varphi}^{(n)}$ as the first ${\widehat{r}_n}$ columns of $\widehat{U}\widehat{\Sigma}$ and $\widehat{\Gamma}_n$ as the first ${\widehat{r}_n}$ columns of $\widehat{V}$ where $\widehat{U}$, $\widehat{\Sigma}$ and $\widehat{V}$ are the terms from \eqref{eqn:SVD} of the singular value decomposition of $Y_{(n)} - X_{(n)}^\prime \beta$. 
Because this error term is a function of $\beta$, an iteration is required between estimating $\beta$ and finding the singular value decomposition of the error term. 
This is a well studied iteration; for details see \cite{Bai2009} or \cite{MoonWeidner2015}. 

In the following assumptions let $\widehat{r}_n$ be the estimated number of factors for the $(n)$-flattening of the regression line when applying the least square methods in \eqref{eqn:factorOptimiser}. 
Also, let $\mathcal{L}\subset \{1,\dots,d\}$ be a non-empty subset of the dimensions. 
In the following, the multilinear rank of $\boldsymbol{\mathcal{A}}$ is restricted such that it is low-rank along at least one of the flattenings.

\begin{assumption}[Bounded norms of covariates and exogenous error]~
\label{ass:norms}
\begin{enumerate}[(i).]
    \item $\big\|{{X}_k}\big\|_F = O_p\left(\prod_{n=1}^d \sqrt{N_n}\right)$ for each $k$
    \item $\norm{{\varepsilon}_{(n^*)}} = O_p\left(\max\{\sqrt{N_{n^*}},\prod_{m\neq n^* }\sqrt{N_m}\}\right)$ for each $n^*\in\mathcal{L}$
\end{enumerate}
\end{assumption}

\begin{assumption}[Weak exogeneity]~
\label{ass:Exog}
    ${\rm vec}(X_k)^\prime{\rm vec}(\varepsilon) = O_p\left(\prod_{n=1}^d \sqrt{N_n}\right)$ for each $k$
\end{assumption}

\begin{assumption}[Low multilinear rank]~
\label{ass:multiLrank}
    For some positive integer, $c$, $r_{n^*} < c $ for all $n^*\in\mathcal{L}$,
    where $r_n$ is the $n$\textsuperscript{th} component of the multilinear rank of $\boldsymbol{\mathcal{A}}$. 
\end{assumption}

\begin{assumption}[Non-singularity]~
\label{ass:NonSing}
    Let $\sigma_s(A)$ be the $s$\textsuperscript{th} singular value for a matrix $A$. 
    Consider linear combinations $\delta_{n^*}\cdot X_{(n^*)} = \sum_k \delta_{n^*,k} X_{(n^*),k} $. 
    For each dimension $n^*\in\mathcal{L}$ that satisfies Assumption~\ref{ass:multiLrank}, then for $K\times 1$ unit vector $\delta_{n^*}$, 
    \begin{align*}
        \min_{\{\delta_{n^*}\in \mathbb{R}^K, \norm{\delta_{n^*}} = 1 \}} 
        \sum_{s = r_{n^*} + \widehat{r}_{n^*} + 1}^{\min\{N_{n^*}, \prod_{m \neq n^*}N_m\}} 
        \sigma_s^2\left( \frac{(\delta_{n^*}\cdot X_{(n^*)})}{\prod_{n}\sqrt{N_n}} \right) > b > 0 \quad \quad wpa1.
    \end{align*}
\end{assumption}

Assumptions~\ref{ass:norms}, \ref{ass:Exog} and \ref{ass:NonSing} are standard regularity assumptions already well established in the literature, e.g. see \cite{MoonWeidner2015}. 
Assumption~\ref{ass:norms}.(i) ensures that the covariates have bounded norms, for example having bounded second moments.
Assumption~\ref{ass:norms}.(ii) allows for some weak correlation across dimensions, see \cite{MoonWeidner2015}, or is otherwise implied if the noise terms are independently distributed with bounded fourth moments, see \cite{latala2005some}. 
Assumption~\ref{ass:Exog} is implied if $X_{i_1,i_2,\dots,i_d;k}\varepsilon_{i_1,i_2,\dots,i_d}$ are zero mean, bounded second moment and only admits weak correlation across dimensions for each $k = 1,\dots, K$. 
Assumption~\ref{ass:NonSing} states that, after factor projection, the set of covariates still collectively admit full-rank variation.

Assumption~\ref{ass:multiLrank} is new and asserts that there exists at least one flattening of the interactive term, $\boldsymbol{\mathcal{A}}$, that is low-dimensional or simply low-rank. 
This requires that at least one of the unobserved terms $\varphi^{(n)}$ is low dimensional. 
Note that not all dimensions must satisfy Assumption~\ref{ass:multiLrank} for the below result. 
If the correct dimension is chosen then variation from the interactive term can be sufficiently projected out using the factor model approach. 
This makes up the statement of the following Proposition. 

\begin{proposition}
\label{prop:bai}
    Let $\widehat{\beta}_{(n^*)}^{2D}$ be the estimator from \cite{Bai2009} after first flattening along dimension $n^*\in\mathcal{L}$. 
    If Assumptions~\ref{ass:norms}-\ref{ass:NonSing} hold, the subset $\mathcal{L}$ is non-empty, and the estimated number of factors $\widehat{r}_{n^*}\geq r_{n^*}$, then, for each $n^*\in\mathcal{L}$ satisfying Assumption~\ref{ass:multiLrank},  
    \begin{align*}
        \norm{\widehat{\beta}_{(n^*)}^{2D} - \beta^0} = O_p\left(\frac{1}{\sqrt{\min\{N_{n^*}, \prod_{n\neq n^*}N_n\}}}\right).
    \end{align*}
\end{proposition}

Proposition~\ref{prop:bai} follows directly from \cite{MoonWeidner2015} since the flattening procedure reduces the problem to the standard linear interactive fixed-effects model. 
This result only applies to estimates in the dimension(s) that satisfy the low-rank assumption in Assumption~\ref{ass:multiLrank}, i.e., the analyst has chosen the correct dimension to flatten. 
The constraint $\widehat{r}_{n^*}\geq r_{n^*}$ can also be changed to $\widehat{r}_{n^*}\geq c$; however, this is more conservative than required for the statement of the result. 
This constraint does not require knowledge of $r_{n^*}$, just that the number of estimated factors is greater than or equal the true number. 
The estimation procedure from Proposition~\ref{prop:bai} can also be augmented to flatten over multiple indices, but makes Assumption~\ref{ass:multiLrank} harder to satisfy. 
For example, take the tensor $\boldsymbol{\mathcal{A}}$ flattened over the first two indices as $\mathcal{A}_{(1,2)}\in \mathbb{R}^{N_1N_2\times\prod_{n\notin\{1,2\}}N_n}$. 
If the parameters $\varphi^{(n)}$ for $n = 3,\dots,d$ are high-dimensional, Assumption~\ref{ass:multiLrank} is only satisfied when both $\varphi^{(1)}$ and $\varphi^{(2)}$ and their product space is low-dimensional. 
However, flattening along multiple dimensions can improve the convergence rate in Proposition~\ref{prop:bai} to 
$O_p\left(\min\{N_{1}N_{2}, \prod_{n\notin\{1,2\}}N_n\}\right)^{-1/2}$.

\subsection{Weighted-within estimator}\label{sect:estimatorsKernel}

Presented here is a simplified version of the estimator, where weights are formed from normed difference over a whole vector. 
This presents a curse of dimensionality that is solved with an iterative version of the estimator. 
Indeed, appendix~\ref{sect:AppendixIterative} presents an iterative version of the estimator that is theoretically more tractable, albeit more complicated. 
The iterated version is a backfitting version of the estimator presented here. 
Simplifying assumptions are made in this section to avoid complexities related to backfitters.

Let $\widehat{\varphi}^{(n)}_{i_n}\in\widehat{\Phi}_n$ generically denote a known or estimated proxy for fixed-effect ${\varphi}^{(n)}_{i_n}$.
Let $\mathcal{W}$ be an ordered set of weight matrices, where the $n$\textsuperscript{th} item ${W}_{n}\in \mathbb{R}^{N_n \times N_n}$ has elements, 
\begin{align}\label{eqn:kernelWeights}
	W_{n,i_ni^\prime_n} := 
    \frac{k\left(\frac{1}{h_n} \norm{\widehat{\varphi}^{(n)}_{i_n} - \widehat{\varphi}^{(n)}_{i^\prime_n}} \right)}
    {\sum_{i^\prime_n = 1}^{N_n} k\left(\frac{1}{h_n} \norm{\widehat{\varphi}^{(n)}_{i_n} - \widehat{\varphi}^{(n)}_{i^\prime_n}}\right)}, 
\end{align}
where $k$ is a kernel function, and $h_n$ is a bandwidth parameter. 
The weighted-within transformation in \eqref{withinGroupEstimator} can be generalised with the following series of $n$-mode products, 
\begin{align*}
    \check{\mathbf{Y}}:=  \mathbf{Y}\times_1 M_1 \times_2 M_2\times_3\dots\times_d M_d, 
\end{align*}
likewise also on each $\mathbf{X}_k$, where $M_n = \mathbb{I}_{N_n} - W_{n}$ and $\times_n$ is the $n$-mode product defined in Section~\ref{sect:notation}. 
The weighted means $\mathbf{Y} - \check{\mathbf{Y}}$ estimate the nonparametric component of $\mathbf{Y}$ that depends on the fixed-effects, and $\check{\mathbf{Y}}$ is the multi-way difference that removes it along every dimension. 
Define the weighted-within estimator as, 
\begin{align*}
    \widehat{\beta}_{\mathcal{W}}: = 
    \left(\frac{1}{\prod_n N_n}\sum_{i_1}\dots\sum_{i_d}\check{X}_{i_1,\dots ,i_d}\check{X}_{i_1,\dots, i_d}^\prime\right)^{-1}
    \frac{1}{\prod_n N_n}\sum_{i_1}\dots\sum_{i_d}\check{X}_{i_1,\dots ,i_d}\check{Y}_{i_1,\dots, i_d}.
\end{align*}

Proxies, $\widehat{\varphi}^{(n)}_{i_n}$, estimated from matrix method in Section~\ref{sect:Matrixestimation}  are considered for inference in Section~\ref{sect:inference}, and is a more challenging setting than when proxies are observed.
Further discussion of proxy estimation is relegated to Section~\ref{sect:DiscussionProxyEst}.

\begin{assumption}[Kernels]\label{ass:Kernels}
Kernel function, $k(\cdot)$,
\begin{enumerate*}[(i).,series = tobecont, itemjoin = \quad]
    \item $k(u)\geq 0$
    \item $\int uk(u)du = 0$\\
    \item $\int k(u)du = 1$
    \item 
    $\int u^2 k(u)du <\infty$
    \item $k(u) = 0$ for $u>U$, $U$ bounded
\end{enumerate*}
\end{assumption}
Assumption~\ref{ass:Kernels}.(i)-(iv) are standard kernel function restrictions. 
Assumption~\ref{ass:Kernels}.(v) regularises the kernel function to decay sufficiently quickly.
The simulations and application use the Gaussian kernel, which does not have compact support, but whose tails are negligible beyond a few bandwidths.

\begin{assumption}[Regularity conditions]~
\label{ass:regCondKer}
Let $\check{T}_{i_1,\dots ,i_d}$ be the entries of tensor $\mathbf{T}$ 
after the weighted fixed-effects are differenced out.
Then, where $f_{\widehat{\Phi}_n}(\widehat\varphi_n)$ are proxy densities,
\begin{enumerate}[(i).]
    \item $\left(\frac{1}{\prod_n N_n}\sum_{i_1}\dots\sum_{i_d}\check{X}_{i_1,\dots ,i_d}\check{X}_{i_1,\dots, i_d}^\prime\right) $ converges in probability to a nonrandom positive definite matrix as
    $N_1,\dots, N_d \rightarrow \infty$. 
    \item $\frac{1}{\prod_n N_n}\sum_{i_1}\dots\sum_{i_d}\check{X}_{i_1,\dots ,i_d}\check{\varepsilon}_{i_1,\dots, i_d}= O_p\left(\frac{1}{\sqrt{\prod_nN_n}}\right)$. 
    \item 
    For all $n = 1,\dots, d$:
    $f_{\widehat{\Phi}_n}(\widehat{\varphi}_{i_{n}}^{({n})}) >0 \,\,\forall \widehat{\varphi}_{i_{n}}^{({n})}\in\widehat\Phi_n$;
    $\sup|f_{\widehat{\Phi}_n}''| <\infty$.
    \item For all $n = 1,\dots, d$: $\max_{j_n}\sum_{i_n = 1}^{N_n} W_{n,i_nj_n} = O_p(1)$.
\end{enumerate}
\end{assumption}
Assumption~\ref{ass:regCondKer}.(i) is analogous to Assumption~\ref{ass:NonSing}, appropriated to the weighted-within projection. 
This can be verified if $\mathbf{X}$ admits a component with sufficiently independent variation across $i_1,\dots,i_d$, e.g. an i.i.d. noise component. 
Assumption~\ref{ass:regCondKer}.(ii) requires weak exogeneity in the covariates after the weighted-within transformation, which can be viewed as similar to Assumption~\ref{ass:Exog}. 
Assumption~\ref{ass:regCondKer}.(ii) is achieved with weakly dependent distributed errors and a sample split, similar to \cite{freeman2023linear}, see Section~\ref{sect:AppendixSupp} including Lemma~\ref{lemma:regCondKer}. 
It is, however, a simplifying assumption in view of the curse of dimensionality when kernel functions are evaluated over normed vectors, not scalars. 
Assumption~\ref{ass:regCondKer}.(ii) could be stated as $O_p(\prod_nN_n^{-3/4}h_n^{-\widehat{r}_n/2})$ to account for the $\widehat{r}_n$-dimensional proxy vector.
However, the backfitting approach in Appendix~\ref{sect:AppendixIterative} alleviates this succinctly, so this complication is ignored here.
Assumption~\ref{ass:regCondKer}.(iv) bounds the total weight that any one unit receives, which keeps each $M_n$ bounded in spectral norm.

\begin{proposition}[Upper bound on kernel weighted estimator]\label{lemma:consKernel}
    Let Assumptions~\ref{ass:Kernels}-\ref{ass:regCondKer} hold, and let $\frac{1}{\prod_n N_n}\|\boldsymbol{\mathcal{A}}\|_F^2 = O_p(1)$.
    For each $n$, write $\mathcal{A}_{(n)} = \varphi^{(n)}\Gamma_n^\prime$ with $\varphi^{(n)}\in\mathbb{R}^{N_n\times r_n}$ as in Section~\ref{sect:Matrixestimation}, and let $\varphi^{(n)}_{i_n}$ be the rows of $\varphi^{(n)}$.
    Let $\mathcal{M}^\prime$ be a non-empty subset of dimensions such that, for each $n^*\in\mathcal{M}^\prime$, the smallest eigenvalue of $\frac{1}{N_{n^*}}\varphi^{(n^*)\prime}\varphi^{(n^*)}$ is bounded away from zero wpa1 and
    $\frac{1}{N_{n^*}}\sum_{i_{n^*}}\left\| \varphi_{i_{n^*}}^{({n^*})} - \widehat{\varphi}_{i_{n^*}}^{({n^*})} \right\|^2 = O_p(C_{n^*}^{-2})$ with $C_{n^*}^{-2}\rightarrow 0$.
    Then, with $N_nh_n \rightarrow \infty$ for all $n$,
    \begin{align*}
        \norm{\widehat{\beta}_{\mathcal{W}} - \beta^0}
        =
        O_p\left(\prod_{n^*\in\mathcal{M}^\prime}
        \left(C_{n^*}^{-1} + h_{n^*}\right)
        \right)
        +
        O_p\left(\prod_{n = 1}^d\frac{1}{\sqrt{N_{n}}} \right)
        .
    \end{align*}
    For $h_n = O(C_n^{-1})$ this reduces to
    \begin{align*}
        \norm{\widehat{\beta}_{\mathcal{W}} - \beta^0}
        =
        O_p\left(\prod_{n^*\in\mathcal{M}^\prime}
        C_{n^*}^{-1}
        \right)
        +
        O_p\left(\prod_{n = 1}^d\frac{1}{\sqrt{N_{n}}} \right)
        .
    \end{align*}
\end{proposition}

Proposition~\ref{lemma:consKernel} shows convergence for the kernel estimator is bounded by convergence of the proxy estimates. 
Smaller bandwidths $h_n$ lead to higher finite sample variance, and also make Assumption~\ref{ass:regCondKer}.(i)-(ii) more difficult to satisfy, hence it is set no smaller than $O(C_n^{-1})$. 
Conditions in Section~\ref{sect:inference}, and discussion in Section~\ref{sect:DiscussionProxyEst} suggest $O_p(C_{n^*}^{-1})$ can be $O_p\big(1/\sqrt{N_{n^*}}\big)$ under regularity conditions imposed in \cite{Bai2009}. 
Hence, the parametric rate is attainable if $\mathcal{M}^\prime = \{1,\dots,d\}$.
The bound does not depend on the number of terms $L$ in \eqref{eqn:fullModel}, and it requires no accuracy of the proxies in dimensions outside $\mathcal{M}^\prime$. 
Implicit in the conditions for $C_{n^*}^{-1} = 1/\sqrt{N_{n^*}}$ is that the correct multilinear rank of the interactive fixed-effects is known, at least for the dimensions $\mathcal{M}^\prime$. 
This can likely be relaxed to the case where the upper bound on the multilinear rank is known, which would follow from the results in \cite{MoonWeidner2015} used in Proposition~\ref{prop:bai}, but left for further research.

\begin{corollary}[Robustness to an unknown low-rank dimension]\label{cor:robust}
    Let the assumptions of Proposition~\ref{lemma:consKernel} hold for some non-empty $\mathcal{M}^\prime$ with $C_n = \sqrt{N_n}$ for $n\in\mathcal{M}^\prime$, and let $h_n = O(N_n^{-1/2})$ with $N_nh_n\rightarrow\infty$ for all $n$.
    Then
    \begin{align*}
        \norm{\widehat{\beta}_{\mathcal{W}} - \beta^0} = O_p\Big(\prod_{n\in\mathcal{M}^\prime}\frac{1}{\sqrt{N_n}}\Big).
    \end{align*}
\end{corollary}

Corollary~\ref{cor:robust} is the main robustness property of the weighted-within estimator.
Proxies are estimated and bandwidths are set in the same way for every dimension, so the estimator does not depend on which dimensions make up $\mathcal{M}^\prime$.
The analyst therefore does not need to know in which dimensions the proxies converge, for example in which flattenings the rank $r_n$ exceeds the number of estimated factors $\widehat{r}_n$, so that the proxies do not converge.
For $\mathcal{M}^\prime = \{n^*\}$ the rate is $N_{n^*}^{-1/2}$, which is the rate of the matrix method in Proposition~\ref{prop:bai} flattened along $n^*$ under the rate restriction of Section~\ref{sect:motivation}, but the matrix method requires the analyst to choose that flattening.
The fixed sample simulations in Section~\ref{sect:sims} illustrate this with multilinear rank $\mathbf{r} = (1, N_1, N_1)$ and one estimated factor, where the matrix method is biased unless the first dimension is flattened.
The rate improves with each further dimension in which the proxies converge, up to the parametric rate for $\mathcal{M}^\prime = \{1,\dots,d\}$, which the asymptotic normality results in Section~\ref{sect:inference} use.

\subsection{Neyman Orthogonal Estimator}\label{sect:inference}

This section shows that the weighted-within estimator is Neyman orthogonal and establishes its asymptotic normality.\footnote{The procedures here may not in general translate to the matrix low-rank approximation estimator, since convergence rates can be too slow for that estimator. }
Proposition~\ref{lemma:consKernel} establishes an upper bound on the kernel weighted fixed-effect estimator convergence rate that is at best exactly the parametric rate for fixed-effect approximation error. 
To ensure the bias from the fixed-effect term converges sufficiently quickly a Neyman orthogonality property is studied here.  

The estimator is presented first as an infeasible estimator, which the weighted-within transformation makes feasible.
Consider the conditional expectation, $\mathbb{E}\left(\boldsymbol{Y}|\boldsymbol{\mathcal{A}}\right) = \boldsymbol{\Gamma}_{Y}$, $\mathbb{E}\left(\boldsymbol{X}|\boldsymbol{\mathcal{A}}\right) = \boldsymbol{\Gamma}_{X}$, such that, 
\begin{align}\label{eqn:maininference}
    \boldsymbol{X} = \boldsymbol{\Gamma}_{X} + \boldsymbol\eta,
    \quad\quad 
    \boldsymbol{\Gamma}_{Y} = \boldsymbol{\Gamma}_{X}\cdot\beta^0 + \boldsymbol{\mathcal{A}},
    \quad\quad 
    \boldsymbol{Y} - \boldsymbol{\Gamma}_{Y} = \sum_{k=1}^K(\boldsymbol{X}_k - \boldsymbol{\Gamma}_{X_k})\cdot\beta_k^0 + \boldsymbol{\varepsilon},
\end{align}
with $\mathbb{E}\left(\boldsymbol\eta|\boldsymbol{\mathcal{A}}\right) = 0$, $\mathbb{E}\left(\boldsymbol{\varepsilon}|\boldsymbol{X},\boldsymbol{\mathcal{A}}\right) = 0$. 
For purposes of estimators considered herein, some regularity in how $\boldsymbol{\mathcal{A}}$ enters the generating process for $\boldsymbol{X}$ is required.
However, as it stands, the display in \eqref{eqn:maininference} is a generic representation, not an imposed model. 
If $\boldsymbol{X}$ is unrelated to $\boldsymbol{\mathcal{A}}$, $\boldsymbol{\Gamma}_{X} = 0$. 
If the correlation between $\boldsymbol{X}$ and $\boldsymbol{\mathcal{A}}$ is weak in the sense that the empirical mean of $\boldsymbol{\Gamma}_{X_k}^2$ converges to zero at arbitrary rate, then the following orthogonality conditions are not required and the convergence result in Proposition~\ref{lemma:consKernel} is sufficient. 
In this sense, the estimation approach below is robust to the existence of $\boldsymbol{\Gamma}_{X_k}$, and how prevalent it is in the generating process for each $\boldsymbol{X}_k$ up to weak regularity conditions. 
The term $\boldsymbol\eta$ must admit sufficient variation, i.e. positive sum of squares, and $\boldsymbol{\Gamma}_{X}$ does need sufficient regularity such that it can be nonparametrically differenced.

The moment conditions used to estimate $\beta$ are, 
\begin{align*}
\mathbb{E}[m(\Gamma_Y,\Gamma_X,\beta)]
=
    \mathbb{E}[(X - \Gamma_X)(Y - \Gamma_Y - (X - \Gamma_X)'\beta)] = 0 
\end{align*}
These conditions are Neyman Orthogonal with respect to $\Gamma_Y,\Gamma_X$ in that, 
\begin{align*}
    \mathbb{E}[\partial m(\Gamma_Y,\Gamma_X,\beta)/\partial \Gamma_Y]
    &= 
    -\mathbb{E}[X - \Gamma_X] = 0
    \\
    \mathbb{E}[\partial m(\Gamma_Y,\Gamma_X,\beta)/\partial \Gamma_{X_k}]
    &= 
    - \mathbb{E}[(Y - \Gamma_Y - (X - \Gamma_X)'\beta)] +  \mathbb{E}[X - \Gamma_X]'\beta = 0
\end{align*}
This property allows error to be multiplicative over estimation error in $\Gamma_Y,\Gamma_X$ in the feasible version of the estimator, since only second-order terms appear in the error. 

An important property of using a nonparametric estimator along with this orthogonality condition is that it admits flexible variation in how fixed-effects enter the equation for $\boldsymbol{X}$. 
Whilst the parametric form for fixed-effects in \eqref{eqn:IntroTuckerModel} and \eqref{eqn:IntroModel} suggest that a nonparametric estimator is not necessary, this logic does not necessarily follow for how fixed-effects enter the structural equation for $\boldsymbol{X}$. 
Indeed, unless unobserved heterogeneity enters the equation for $\boldsymbol{X}$ in a strict parametric form as per the common correlated effects model, some form of additional debias would be necessary for sufficient decay in asymptotic bias. 
Henceforth, nonparametric estimators are implemented, shown to be sufficient for the parametric form in \eqref{eqn:IntroTuckerModel} and \eqref{eqn:IntroModel}, and motivated to also control generic variation for how fixed-effects enter $\boldsymbol{X}$. 
How fixed-effects enter $\boldsymbol{X}$ can be flexible, but must still adhere to nonparametric regression regularity conditions.

The (infeasible) inference corrected estimator, $\widehat\beta_{IC}^{(infeasible)}$, is
\begin{align*}
    \widehat\beta_{IC}^{(infeasible)} = \left({\rm vec}_K(\boldsymbol{X} - \boldsymbol{\Gamma}_{X})^\prime {\rm vec}_K(\boldsymbol{X} - \boldsymbol{\Gamma}_{X}) \right)^{-1} \left({\rm vec}_K(\boldsymbol{X} - \boldsymbol{\Gamma}_{X})^\prime {\rm vec}(\boldsymbol{Y} - \boldsymbol{\Gamma}_{Y}) \right),
\end{align*}
where ${\rm vec}_K(\boldsymbol{X} - \boldsymbol{\Gamma}_{X})\in \mathbb{R}^{\prod_n N_n \times K}$ is shorthand for the matrix of vectorised covariates, with each column the vectorised $(\boldsymbol{X}_k - \boldsymbol{\Gamma}_{X_k})$. The ${\rm vec}$ notation for $(\boldsymbol{Y} - \boldsymbol{\Gamma}_{Y})$ is the standard vectorisation. 
This estimator is robust to fixed-effects appearing in the equation for $Y$ and/or $X$, and is less sensitive to misspecification of the fixed-effects in each. 
Terms $\Gamma_Y$ and $\Gamma_X$ do not need to be explicitly estimated for the feasible estimator, only differenced out from $Y$, respectively $X$. 

The weighted-within estimator of Section~\ref{sect:estimatorsKernel} is the feasible version of this estimator. 
Setting $\widehat{\boldsymbol{\Gamma}}_{Y} = \boldsymbol{Y} - \check{\boldsymbol{Y}}$ and $\widehat{\boldsymbol{\Gamma}}_{X_k} = \boldsymbol{X}_k - \check{\boldsymbol{X}}_k$, the weighted means differenced out in Section~\ref{sect:estimatorsKernel}, and replacing $\boldsymbol{\Gamma}_{Y}$ and $\boldsymbol{\Gamma}_{X}$ by these estimates gives exactly $\widehat{\beta}_{\mathcal{W}}$. 
The weighted-within transformation is automatically endowed with the Neyman orthogonal property, since it is simply a linear smoother across both $Y$ and $X$. 
That is, for the linear smoother $S\boldsymbol{T} := \boldsymbol{T} - \check{\boldsymbol{T}}$, $SY = SX\beta^0 + S\mathcal{A} + S\varepsilon =\hat \Gamma_X \beta^0 + \hat{\mathcal{A}} + S\varepsilon = \hat\Gamma_Y$. 
For thorough details of this property, see Section 2 of \cite{freeman2026inference}. 
Write $\widehat\beta_{IC}$ for the feasible estimator with generic estimates $\widehat{\boldsymbol{\Gamma}}_{Y}$ and $\widehat{\boldsymbol{\Gamma}}_{X}$, so that $\widehat\beta_{IC} = \widehat\beta_{\mathcal{W}}$ for the weighted-within estimates, and the results below apply to the weighted-within estimator. 
Algorithm~\ref{alg:feasible} summarises the steps. 

\begin{algo}[Weighted-within estimator]\label{alg:feasible}~
\begin{enumerate}[1.]
    \item For each $n = 1,\dots,d$, compute the estimator of Section~\ref{sect:Matrixestimation} after flattening along dimension $n$ with $\widehat{r}_n$ factors, and take the proxies $\widehat{\varphi}^{(n)}_{i_n}$ as the rows of the leading $\widehat{r}_n$ left singular vectors of that fit, as in Section~\ref{sect:DiscussionProxyEst}. 
    \item Form the weights $W_n$ in \eqref{eqn:kernelWeights} with bandwidths $h_n$, and set $M_n = \mathbb{I}_{N_n} - W_n$. 
    \item Compute $\check{\boldsymbol{Y}} = \boldsymbol{Y}\times_1 M_1 \times_2 \dots \times_d M_d$ and likewise each $\check{\boldsymbol{X}}_k$, so that $\widehat{\boldsymbol{\Gamma}}_{Y} = \boldsymbol{Y} - \check{\boldsymbol{Y}}$ and $\widehat{\boldsymbol{\Gamma}}_{X_k} = \boldsymbol{X}_k - \check{\boldsymbol{X}}_k$. 
    \item Compute $\widehat{\beta}_{\mathcal{W}}$ by least squares of ${\rm vec}(\check{\boldsymbol{Y}})$ on ${\rm vec}_K(\check{\boldsymbol{X}})$, with standard errors from Theorems~\ref{thm:AsyNormHom}--\ref{thm:AsyNormHetACCC}. 
\end{enumerate}
\end{algo}

As is shown below, the weighted-within estimator is robust to fixed-effects appearing in very general form in $X$. Indeed, slow consistency rates for ${\boldsymbol{\Gamma}}_{X}$ estimates are sufficient under weak regularity in $\varepsilon$.
Under dependent errors, Assumption~\ref{ass:inferenceInd} below is met by computing Steps 1 to 3 on the sample split of Appendix~\ref{sect:AppendixSupp}.
The application uses the iterative version in Appendix~\ref{sect:AppendixIterative}, which re-estimates the proxies at each iteration and forms the weights from one scalar proxy at a time to allow for higher-dimensional fixed-effects.

Consider estimators for ${\boldsymbol{\Gamma}}_{X}$, and ${\boldsymbol{\Gamma}}_{Y}$ as $\widehat{\boldsymbol{\Gamma}}_{X}$, and $\widehat{\boldsymbol{\Gamma}}_{Y}$, respectively, and $\widehat{\Omega}_X = N^{-1}{\rm vec}_K(\boldsymbol{X} - \widehat{\boldsymbol{\Gamma}}_{X})^\prime {\rm vec}_K(\boldsymbol{X} - \widehat{\boldsymbol{\Gamma}}_{X})$. 
Then,  with $\widehat{\boldsymbol{\eta}}:= \boldsymbol{X} - \widehat{\boldsymbol{\Gamma}}_{X}$,
\begin{align}
    \widehat{\Omega}_X^{-1}N^{-1}{\rm vec}_K(\widehat{\boldsymbol{\eta}})^\prime 
    &{\rm vec}(\boldsymbol{Y} - \widehat{\boldsymbol{\Gamma}}_{Y}) 
    \nonumber
    = 
    \widehat{\Omega}_X^{-1}N^{-1}{\rm vec}_K(\widehat{\boldsymbol{\eta}})^\prime 
    {\rm vec}(\boldsymbol{X}\cdot\beta^0 + \boldsymbol{\mathcal{A}} + \boldsymbol{\varepsilon} - \widehat{\boldsymbol{\Gamma}}_{X}\cdot{\beta^0} - \widehat{\boldsymbol{\mathcal{A}}})
    \nonumber
    \\
    &=
    \beta^0 + 
    \widehat{\Omega}_X^{-1}N^{-1}
    {\rm vec}_K(\widehat{\boldsymbol{\eta}})^\prime
    \big(
    {\rm vec}
    (\boldsymbol{\mathcal{A}}  - \widehat{\boldsymbol{\mathcal{A}}})
    +
    {\rm vec}(\boldsymbol{\varepsilon})
    \big)
    \nonumber
    \\
    &=
    \label{eqn:InfError}
    \beta^0 + \widehat{\Omega}_X^{-1}N^{-1}
    {\rm vec}_K(\widehat{\boldsymbol{\eta}})^\prime{\rm vec}(\boldsymbol{\mathcal{A}}  - \widehat{\boldsymbol{\mathcal{A}}})
    + \widehat{\Omega}_X^{-1}N^{-1}{\rm vec}_K(\widehat{\boldsymbol{\eta}})^\prime{\rm vec}(\boldsymbol{\varepsilon}),
\end{align}
where $\widehat{\boldsymbol{\mathcal{A}}}$ is the preliminary estimator of $\boldsymbol{\mathcal{A}}$ used to form $\widehat{\boldsymbol{\Gamma}}_{Y}$. 
If $\widehat{\Omega}_X^{-1}N^{-1}{\rm vec}_K(\widehat{\boldsymbol{\eta}})^\prime{\rm vec}(\boldsymbol{\mathcal{A}}  - \widehat{\boldsymbol{\mathcal{A}}}) = o_p\left({\prod_{n = 1}{N_n^{-1/2}}}\right)$, then standard asymptotic normality arguments follow. 
Hence, convergence rates on $({\boldsymbol{\Gamma}}_{X}
- \widehat{\boldsymbol{\Gamma}}_{X})$, and $(\boldsymbol{\mathcal{A}}  - \widehat{\boldsymbol{\mathcal{A}}})$ are studied. 

From Section~\ref{sect:estimatorsKernel}, with the weighted estimator, $\Big(\big(\prod_{n = 1}{N_n^{-1}}\big){\rm vec}(\boldsymbol{\mathcal{A}}  - \widehat{\boldsymbol{\mathcal{A}}})^\prime {\rm vec}(\boldsymbol{\mathcal{A}}  - \widehat{\boldsymbol{\mathcal{A}}})\Big)^{1/2} = O_p(\prod_n h_n)$, where $h_n \gtrsim N_n^{-1/2}$ is the bandwidth used for dimension $n$. 
Define, 
\begin{align}\label{eqn:RateDefn}
    \frac{1}{N}{\rm vec}_K({\boldsymbol{\Gamma}}_{X} - \widehat{\boldsymbol{\Gamma}}_{X})^\prime{\rm vec}(\boldsymbol{\mathcal{A}}  - \widehat{\boldsymbol{\mathcal{A}}}) = O_p(\xi_{X\mathcal{A}}),
    \\ \nonumber
    \frac{1}{N}\Big\|{\rm vec}_K({\boldsymbol{\Gamma}}_{X} - \widehat{\boldsymbol{\Gamma}}_{X})^\prime{\rm vec}_K({\boldsymbol{\Gamma}}_{X} - \widehat{\boldsymbol{\Gamma}}_{X})\Big\|_F^2 = O_p(\xi_{{X}}^2),
    &&
    \frac{1}{N}{\rm vec}(\boldsymbol{\mathcal{A}}  - \widehat{\boldsymbol{\mathcal{A}}})^\prime{\rm vec}(\boldsymbol{\mathcal{A}}  - \widehat{\boldsymbol{\mathcal{A}}}) = O_p(\xi_{\mathcal{A}}^2). 
\end{align}
Under Assumption~\ref{ass:regCondKer}, $\widehat{\Omega}_X^{-1} = O_p(1)\mathbbm{1}_{K\times K}$, and iid $\boldsymbol\varepsilon, \boldsymbol\eta$ it can be shown,
\begin{align*}
    \sqrt{N}(\widehat{\beta}_{IC} -\beta^0) = 
    \sqrt{N}\left(
    O_p(\xi_{X\mathcal{A}})
    + N^{-1/2}O_p(\xi_{\mathcal{A}} + \xi_{{X}})
    + \widehat{\Omega}_X^{-1}N^{-1}{\rm vec}_K({\boldsymbol{\eta}})^\prime{\rm vec}(\boldsymbol{\varepsilon})
    \right),
\end{align*}
and under weak dependence structures in $\boldsymbol\varepsilon, \boldsymbol\eta$, 
\begin{align*}
    \sqrt{N}(\widehat{\beta}_{IC} -\beta^0) = 
    \sqrt{N}\left(
    O_p(\xi_{X\mathcal{A}})
    + N^{-1/4}O_p(\xi_{\mathcal{A}} + \xi_{{X}})
    + \widehat{\Omega}_X^{-1}N^{-1}{\rm vec}_K({\boldsymbol{\eta}})^\prime{\rm vec}(\boldsymbol{\varepsilon})
    \right).
\end{align*}
By Cauchy-Schwarz  $\xi_{X\mathcal{A}} \leq \xi_\mathcal{A}\xi_{{X}}$ which is $O_p(N^{-1/2}\xi_{{X}})$ by Proposition~\ref{lemma:consKernel} proof. 
Hence, under iid noise terms, $O_p(\xi_{X\mathcal{A}}) = o_p(N^{-1/2})$ and $\{\xi_{\mathcal{A}}, \xi_{{X}}\} = o_p(1)$, asymptotic bias is sufficiently small. 
Under weak dependence, $O_p(\xi_{X\mathcal{A}}) = o_p(N^{-1/2})$, and $\{\xi_{\mathcal{A}}, \xi_{{X}}\} = o_p(N^{-1/4})$, is required. 
Proof of Proposition~\ref{lemma:consKernel} establishes the upper bound on convergence for $\xi_{\mathcal{A}}$ can be $O_p(N^{-1/2})$, such that for any $\xi_X = o_p(N^{-1/4})$, convergence rate for asymptotic bias is $o_p(N^{-1/2})$. 

Dependence between either $\boldsymbol\eta$, respectively $\boldsymbol{\varepsilon}$, and $\widehat{\boldsymbol{\Gamma}}_{X}$ or $\widehat{\boldsymbol{\mathcal{A}}}$ can occur because the estimator $\widehat{\boldsymbol{\Gamma}}_{X}$, and/or $\widehat{\boldsymbol{\mathcal{A}}}$, may be functions of $\boldsymbol\eta$, respectively $\boldsymbol{\varepsilon}$. 
To break this dependence, a simple sample splitting procedure can be implemented, 
e.g., as outlined in \cite{freeman2023linear}, which is transferable to the independently distributed error setting. 

Below regularity conditions from \cite{Bai2009,BaiNg2002} ensure estimates of the fixed-effects used as proxies in the kernel weights converge at the sufficient rate. 

\begin{assumption}[\cite{Bai2009} conditions]\label{ass:bai2009}
Let $N:= \prod_{n = 1}{N_n}$. 
\begin{enumerate*}[(i).,series = tobecont, itemjoin = \quad]
    \item $\mathbb{E}\|X_{i_1\dots i_d}\|^4$ bounded.
    \\
    \item For each $n=1,\dots, d$, $\mathbb{E}\|\varphi_{i_n}^{(n)}\|^4$ bounded, and $\frac{1}{N_n}\sum_{i_n = 1}^{N_n}\varphi_{i_n}^{(n)}\varphi_{i_n}^{(n)\,\prime}$ converges to an $r_n\times r_n$ p.d. matrix as $N_n\rightarrow\infty$. 
    \item $\mathbb{E}\varepsilon_{i_1\dots i_d} = 0$, $\mathbb{E}(\varepsilon_{i_1\dots i_d})^{8}$ bounded. 
    \\
    \item $\mathbbm{E}\varepsilon_{i_1^{\,}\dots i_d^{\,}}
            \varepsilon_{i_1^{\prime}\dots i_d^{\prime}} := \tau_{i_1^{\,}\dots i_d^{\,}; i_1^{\prime}\dots i_d^{\prime}}$, and 
            $N^{-1}\sum_{i_1^{\,}\dots i_d^{\,}} \sum_{i_1^{\prime}\dots i_d^{\prime}} |\tau_{i_1^{\,}\dots i_d^{\,}; i_1^{\prime}\dots i_d^{\prime}}|$ bounded. 
    \item $\varepsilon_{i_1\dots i_d}$ independent of $X_{i_1^\prime\dots i_d^\prime}$ and $\varphi_{i_n^\prime}^{(n)}$ for all $i_1\dots i_d$ and $i_1^\prime\dots i_d^\prime$, and $n=1,\dots, d$. 
\end{enumerate*}
\begin{enumerate}[(i).,resume = tobecont, ,topsep = 0pt, partopsep = 0pt]
    \item For $\sigma_{i_n,j_n,k_n,m_n}^{(n)} := cov(\varepsilon_{i_1,\dots,i_n,\dots,i_d}\varepsilon_{i_1,\dots,j_n,\dots,i_d},\varepsilon_{i_1,\dots,k_n,\dots,i_d}\varepsilon_{i_1,\dots,m_n,\dots,i_d})$, for all $n=1,\dots,d$;  
    \begin{align*}
        N_n^{-2}\prod_{n^\prime \neq n} N_{n^\prime}^{-1} \sum_{i_n,j_n,k_n,m_n} \sum_{i_{n^\prime}, n^\prime\neq n} |\sigma_{i_n,j_n,k_n,m_n}^{(n)}| \leq M <\infty
    \end{align*}
\end{enumerate}
\end{assumption}

\begin{corollary}\label{cor:bai2009}
    If Assumption~\ref{ass:bai2009}, and assumptions in Proposition~\ref{prop:bai} hold then: \\ $C_n^{-2} = 1/\min\{N_n, \prod_{n^\prime\neq n} N_{n^\prime}\}$ for each $n = 1,\dots,d$ if the \cite{Bai2009} estimator is used in the matrix low-rank setting from Section~\ref{sect:Matrixestimation} separately for each $n=1,\dots,d$.\footnote{See Proposition A.1 in \cite{Bai2009} appendix}
\end{corollary}

Below are sampling assumptions to achieve asymptotic normality. 
\begin{assumption}\label{ass:inf3}
Let $\xi_X$, and $\xi_{\mathcal{A}}$ be sequences bounded in probability as $N_n\rightarrow\infty$ for each $n= 1,\dots d$. Maintain $N:= \prod_{n = 1}^d N_n$.
As $N_n\rightarrow\infty$ for each $n= 1,\dots d$, $\{\xi_X,\xi_{\mathcal{A}}\} = o_p(N^{-1/4})$, 
$ \xi_{X{\mathcal{A}}} = o_p(N^{-1/2})$, where $\{\xi_X,\xi_{\mathcal{A}},\xi_{X{\mathcal{A}}}\}$ are defined in \eqref{eqn:RateDefn}.
\end{assumption}

Assumption~\ref{ass:inf3} implies there exists a consistent estimate of $\mathbb{E}\left(\boldsymbol{X}|\boldsymbol{\mathcal{A}}\right)$ at rate $N^{-1/4}$. 
For iid idiosyncratic terms $\boldsymbol\varepsilon, \boldsymbol\eta$ this can be relaxed to consistency at arbitrary rate. 
Hence, in the context of Proposition \ref{lemma:consKernel}, this is not an overly strong assumption on $\widehat{\boldsymbol{\Gamma}}_{X_k}$ when errors are iid. 
Indeed, by the proof of Proposition \ref{lemma:consKernel} and Corollary~\ref{cor:bai2009}, convergence rates in Assumption~\ref{ass:inf3} are $\xi_{\mathcal{A}} = O_p(1/\sqrt{N})$. 
Assumption~\ref{ass:inf3} also details a reprieve from requiring $\xi_{\mathcal{A}} = O_p(1/\sqrt{N})$ if indeed $\xi_X$ converges at a faster rate. 

When errors are not iid, it is sufficient to impose additional regularity on $\mathbb{E}\left(\boldsymbol{X}|\boldsymbol{\mathcal{A}}\right)$, such as components of $\boldsymbol{\mathcal{A}}$ entering this equation via a smooth function. 
Since components of $\boldsymbol{\mathcal{A}}$ are identified through the parametric form they enter the equation for $\boldsymbol{Y}$, the condition that estimator $\widehat{\boldsymbol{\Gamma}}_{X_k}$ is consistent at the $N^{-1/4}$ rate follows from standard nonparametric arguments once endowed with this regularity. 
The main contribution of this paper is to provide the novel within transformation and show these satisfy Neyman orthogonality conditions. 
Sufficient conditions for the rate in Assumption~\ref{ass:inf3} include $\mathbb{E}\left(\boldsymbol{X}|\boldsymbol{\mathcal{A}}\right) = g(\varphi^{(1)}, \dots, \varphi^{(d)})$ for twice continuously differentiable $g$ in all arguments, see e.g. \cite{freeman2023linear, freeman2026inference}. 
Equally, a CCE type constraint which imposes $\boldsymbol{\mathcal{A}}$ enters $\mathbb{E}\left(\boldsymbol{X}|\boldsymbol{\mathcal{A}}\right) $ parametrically also provides sufficient regularity.

The next assumption restricts $\boldsymbol\eta$ and $\boldsymbol{\varepsilon}$ to be independent of $\widehat{\boldsymbol{\Gamma}}_{X}$ and $\widehat{\boldsymbol{\mathcal{A}}}$. 
As mentioned already, a sample splitting device can achieve this.\footnote{Appendix~\ref{sect:AppendixSupp} details a sample splitting device that can be used for  $\widehat{\boldsymbol{\Gamma}}_{X}$ and $\widehat{\boldsymbol{\mathcal{A}}}$ in the presence of weakly dependent idiosyncratic terms. }
Let $\independent$ denote stochastic independence. 
\begin{assumption}\label{ass:inferenceInd}
Let $\widehat{\boldsymbol{\Gamma}}_{X}$ be the estimate for $\boldsymbol{\Gamma}_{X}$ and $\widehat{\boldsymbol{\mathcal{A}}}$ the estimate for ${\boldsymbol{\mathcal{A}}}$ in \eqref{eqn:maininference}. Then, for all $k =1,\dots K$,
    \begin{enumerate*}[(i).,series = tobecont, itemjoin = \quad]
        \item $\boldsymbol{\eta}_{k,i_1\dots i_d}\independent 
        {\boldsymbol{\mathcal{A}}}_{i_1\dots i_d},
        \widehat{\boldsymbol{\mathcal{A}}}_{i_1\dots i_d}$;
        \item $\boldsymbol{\varepsilon}_{i_1\dots i_d}\independent \widehat{\boldsymbol{\Gamma}}_{X_k,i_1\dots i_d}$.
    \end{enumerate*}
\end{assumption}
Independence between $\boldsymbol{\eta}_{k,i_1\dots i_d}$ and  ${\boldsymbol{\mathcal{A}}}_{i_1\dots i_d}$ in Assumption~\ref{ass:inferenceInd}.(i) is made for a cleaner derivation, but can be weakly dependent under further regularity conditions on $\boldsymbol{\eta}_{k,i_1\dots i_d}$. 
Assumption~\ref{ass:inferenceInd}.(ii) implies in \eqref{eqn:InfError}, $\widehat{\Omega}_X^{-1}{\rm vec}_K(\widehat{\boldsymbol{\eta}})^\prime{\rm vec}(\boldsymbol{\varepsilon}) = \widehat{\Omega}_X^{-1}{\rm vec}_K({\boldsymbol{\eta}})^\prime{\rm vec}(\boldsymbol{\varepsilon}) + o_p(N^{-1/2})$. 
\begin{assumption}\label{ass:inf2}
    $\mathbb{E}\left[ {\eta}_{i_1\dots i_d}^{\,} {\eta}_{i_1\dots i_d}^{\prime} \right] $ is
non--singular and $E\left[ \varepsilon_{i_1\dots i_d}|{\eta}_{i_1\dots i_d}\right] =0$.
\end{assumption}
With Proposition~\ref{lemma:consKernel} assumptions, Assumptions~\ref{ass:bai2009}-\ref{ass:inf2}, and with $C_n^{-1} \asymp h_n = O_p(1/\sqrt{N_n})$ for all $n = 1,\dots,d$ and $\mathcal{M}^\prime = \{1,\dots,d\}$, maintaining that $N:= \prod_{n = 1}{N_n}$, 
\begin{align*}
    \sqrt{N}(\widehat{\beta}_{IC} -\beta^0) 
    &= 
    o_p(1)
    + 
    \widehat{\Omega}_X^{-1}
    \frac{1}{\sqrt{N}}{\rm vec}_K({\boldsymbol{\eta}})^\prime{\rm vec}(\boldsymbol{\varepsilon})
    \\
    &=
    o_p(1)
    + 
    \left(
    \mathbb{E}\left[ {\eta}_{i_1\dots i_d}^{\,} {\eta}_{i_1\dots i_d}^{\prime} \right] 
    + o_p(1)
    \right)^{-1}
    \frac{1}{\sqrt{N}}{\rm vec}_K({\boldsymbol{\eta}})^\prime{\rm vec}(\boldsymbol{\varepsilon})
\end{align*}

In the presence of both heteroskedasticity and correlation in all dimensions, the most flexible specification to be considered here, a central limit theorem is required for
$(N)^{-1/2}{\rm vec}_K({\boldsymbol{\eta}})^\prime{\rm vec}(\boldsymbol{\varepsilon})$. 
In this case the variance is, 
\begin{align*}
    \textrm{var}\left( 
    \frac{1}{\sqrt{N}}\sum_{i_1\dots i_d} {\eta}_{i_1\dots i_d} \varepsilon_{i_1\dots i_d}
    \right)
    = 
    \frac{1}{{N}}
    \sum_{i_1^{\,},i_1^{\prime}} \dots
    \sum_{i_d^{\,} i_d^{\prime}}
    \mathbbm{E}({\eta}_{i_1\dots i_d}^{\,} {\eta}_{i_1^\prime\dots i_d^\prime}^\prime
    )\mathbbm{E}(
    \varepsilon_{i_1\dots i_d} \varepsilon_{i_1^\prime\dots i_d^\prime} ).
\end{align*}

Stronger independent and identically distributed assumptions are considered first, which are sequentially relaxed for stronger results later in the section. 
\begin{assumption}\label{ass:erroriid}
    \begin{enumerate*}[(i).,series = tobecont, itemjoin = \quad]
    \item $\left\{ \varepsilon _{i_1\dots i_d},{\eta}_{i_1\dots i_d}\right\} $, are i.i.d. across $i_1\dots i_d$,
    \\
    \item $\mathbb{E}\big(\varepsilon _{i_1\dots i_d}^2|{\eta}_{i_1\dots i_d}\big)
    =:\sigma_{\varepsilon}^2 \leq M < \infty $.
    \end{enumerate*}
\end{assumption}

\begin{theorem}[Asymptotic distribution under homoskedasticity]\label{thm:AsyNormHom}
        Let the assumptions in Proposition~\ref{lemma:consKernel} hold.
    Additionally, let Assumptions~\ref{ass:bai2009}-\ref{ass:erroriid} hold. Then, for $N_n\rightarrow\infty$, with $N_n \lesssim \prod_{n^\prime\neq n}N_{n^\prime}$, for all $n = 1,\dots,d$, 
    \begin{align*}
        \sqrt{N}\big(\widehat{\beta}_{IC} -\beta^0\big) 
        \xrightarrow[d]{}
        \mathcal{N}\left(0, \sigma_{\varepsilon}^2\Omega_X^{-1}
        \right),
\quad\quad 
\Omega_X:= \plim_{N\rightarrow\infty} \frac{1}{N}\sum_{i_1\dots i_d} {\eta}_{i_1\dots i_d}^{\,} {\eta}_{i_1\dots i_d}^{\prime}
    \end{align*}
\end{theorem}

Consistent estimation of the variance term is possible under the same set of assumptions. 
Estimate $\widehat{\boldsymbol{\varepsilon}} =  \boldsymbol{Y} - \boldsymbol{X}\cdot\widehat\beta_{IC} - \widehat{\boldsymbol{\mathcal{A}}}$. 
Then, $\frac{1}{N}\sum_{i_1\dots i_d}\widehat{\varepsilon}_{i_1\dots i_d}^2 = \sigma_\varepsilon^2 + o_p(1)$ follows from the consistency of $\widehat{\boldsymbol{\Gamma}}_{Y}$, $\widehat{\boldsymbol{\Gamma}}_{X}$, and $\widehat\beta_{IC}$. 
Likewise $\mathbb{E}\left[ {\eta}_{i_1\dots i_d}^{\,} {\eta}_{i_1\dots i_d}^{\prime} \right]$ can be estimated consistently from the sample analog $\frac{1}{N}\sum_{i_1\dots i_d}
\left(X_{i_1\dots i_d} - \widehat{\Gamma}_{X, _{i_1\dots i_d}}\right)
\left(X_{i_1\dots i_d} - \widehat{\Gamma}_{X, _{i_1\dots i_d}}\right)^\prime
$.

Assumptions on the error can be used to satisfy Lyapunov conditions. Instead, a central limit theorem is assumed, as is common in the interactive fixed-effects literature.\footnote{See for example Assumption E in Section 5 in \cite{Bai2009}. }

\begin{assumption}\label{ass:heteroCLT}
    \begin{enumerate}[(i).]
        \item $\varepsilon _{i_1\dots i_d}$ independent across $i_1\dots i_d$. ${\eta}_{i_1\dots i_d}$ are i.i.d. $i_1\dots i_d$.
        \item $\sigma_{i_1\dots i_d}^2 := \mathbb{E}\big(\varepsilon _{i_1\dots i_d}^2|\boldsymbol{\eta}\big)$ is bounded.
        \item For nonrandom positive definite $\Sigma$, 
        \begin{align*}
            \plim_{N\rightarrow\infty} \frac{1}{N}\sum_{i_1\dots i_d} \sigma_{i_1\dots i_d}^2 
            {\eta}_{i_1\dots i_d}^{\,} {\eta}_{i_1\dots i_d}^{\prime}
            = \Sigma,
            \quad\quad\quad
            \frac{1}{\sqrt{N}}\sum_{i_1\dots i_d} 
            {\eta}_{i_1\dots i_d} \varepsilon_{i_1\dots i_d}
            \xrightarrow[]{d}
            \mathcal{N}(0,\Sigma).
        \end{align*}
    \end{enumerate}
\end{assumption}

\begin{theorem}[Asymptotic distribution under heteroskedasticity]\label{thm:AsyNormHet}
        Let the assumptions in Proposition~\ref{lemma:consKernel} hold.
    Additionally, let Assumptions~\ref{ass:bai2009}-\ref{ass:inf2} and Assumption~\ref{ass:heteroCLT} hold. Then, for $N_n\rightarrow\infty$, with $N_n \lesssim \prod_{n^\prime\neq n}N_{n^\prime}$, for all $n = 1,\dots,d$,
    \begin{align*}
        \sqrt{N}\big(\widehat{\beta}_{IC} -\beta^0\big) 
        \xrightarrow[d]{}
        \mathcal{N}\left(0, \Omega_X^{-1}\,\Sigma\,\Omega_X^{-1}
        \right), 
\quad\quad 
\Omega_X:= \plim_{N\rightarrow\infty} \frac{1}{N}\sum_{i_1\dots i_d} {\eta}_{i_1\dots i_d}^{\,} {\eta}_{i_1\dots i_d}^{\prime}.
    \end{align*}
\end{theorem}

\begin{assumption}\label{ass:HeteroCorrCLT}
        For nonrandom positive definite $\widetilde{\Sigma}$, 
        \begin{align*}
            \plim_{N\rightarrow\infty} \frac{1}{N}
            \sum_{i_1^{\,}\dots i_d^{\,}} \sum_{i_1^{\prime}\dots i_d^{\prime}} 
            \varepsilon_{i_1^{\,}\dots i_d^{\,}}
            \varepsilon_{i_1^{\prime}\dots i_d^{\prime}}
            {\eta}_{i_1^{\,}\dots i_d^{\,}}^{\,} {\eta}_{i_1^{\prime}\dots i_d^{\prime}}^{\prime}
            = \widetilde{\Sigma},
            \quad\quad\quad
            \frac{1}{\sqrt{N}}\sum_{i_1\dots i_d} 
            {\eta}_{i_1\dots i_d} \varepsilon_{i_1\dots i_d}
            \xrightarrow[]{d}
            \mathcal{N}(0,\widetilde{\Sigma}).
        \end{align*}
\end{assumption}
Under correlation, independence in Assumption~\ref{ass:inferenceInd} needs a stronger sample split than from \cite{freeman2023linear} -- a novel sample split device is devised in Appendix~\ref{sect:AppendixSupp}. 
\begin{assumption}\label{ass:infHighLevel}
    Define 
        $\sigma_{i_1\dots i_d; i_1^\prime\dots i_d^\prime}
        :=
        \mathbbm{E}(\varepsilon_{i_1\dots i_d} \varepsilon_{i_1^\prime\dots i_d^\prime} )$ 
        and 
        $\sigma_{k;i_1\dots i_d; i_1^\prime\dots i_d^\prime}
        :=
        \mathbbm{E}(\eta_{k:i_1\dots i_d} \eta_{k:i_1^\prime\dots i_d^\prime} )$.
        For finite $M<\infty$ such that, $\plim_{N\rightarrow\infty} \xi_{\sigma,N}^2
         \leq M$, $\plim_{N\rightarrow\infty} \xi_{\sigma_k,N}^2
        \leq M$ where,   
        $$\xi_{\sigma,N}^2 := \frac{1}{N} \sum_{i_1^\prime\dots i_d^\prime\neq i_1\dots i_d}
        \sum_{i_1\dots i_d}
        \sigma_{i_1\dots i_d; i_1^\prime\dots i_d^\prime}^2 ,
        \quad \quad
        \xi_{\sigma_k,N}^2 := \frac{1}{N} \sum_{i_1^\prime\dots i_d^\prime\neq i_1\dots i_d}
        \sum_{i_1\dots i_d}
        \sigma_{k;i_1\dots i_d; i_1^\prime\dots i_d^\prime}^2.
        $$
        For $\xi_X,\xi_{\mathcal{A}}$ defined in Assumption~\ref{ass:inf3}, $\xi_X^2 \xi_{\sigma,N} N^{1/2} = o_p(1)$, $\xi_{\mathcal{A}}^2 \xi_{\sigma_k,N} N^{1/2} = o_p(1)$ for all $k$.
\end{assumption}
Assumption~\ref{ass:infHighLevel} can be achieved if, e.g. $\{\xi_X,\xi_\mathcal{A}\}  = o_p(N^{-1/4})$, without further restrictions on $\xi_{\sigma,N}$. If $\varepsilon_{i_1\dots i_d}$ are i.i.d., or i.n.i.d., $\xi_{\sigma,N}=0$, satisfying the assumption. 

\begin{theorem}[Asymptotic distribution under heteroskedasticity and correlation]\label{thm:AsyNormHetACCC}
    Let the assumptions in Proposition~\ref{lemma:consKernel} hold. Additionally, let Assumptions~\ref{ass:bai2009}-\ref{ass:inf2}, and \ref{ass:HeteroCorrCLT}-\ref{ass:infHighLevel} hold. 
    Then, for $N_n\rightarrow\infty$, with $N_n \lesssim \prod_{n^\prime\neq n}N_{n^\prime}$, for all $n = 1,\dots,d$,
    \begin{align*}
        \sqrt{N}\big(\widehat{\beta}_{IC} -\beta^0\big) 
        \xrightarrow[d]{}
        \mathcal{N}\left(0, \Omega_X^{-1}\,\widetilde{\Sigma}\,\Omega_X^{-1}
        \right).
    \end{align*}
\end{theorem}
From \cite{Bai2009}, $\widetilde{\Sigma}$ can be estimated with the \cite{newey1987hac} kernel approach for time dimension, in conjunction with a partial sample estimator for cross-sections.

\section{Discussion and Implementation}\label{sect:Discussion}

This section serves to discuss the results in Section~\ref{sect:estimators}, motivate further some of the chosen methods, and provide some methods to estimate proxies for kernel weights. 

\subsection{Matrix method results}\label{sect:DiscussionMatrix}

Proposition~\ref{prop:bai} takes for granted the dimension to flatten across admits a low rank interactive fixed-effect term for the least square method in \cite{Bai2009}. 
Established diagnostics in \cite{BaiNg2002}, \cite{ahn2013eigenvalue} and \cite{hallin2007determining} can be used to determine the number of factors in the pure factor model, which could be repeated along the different flattenings in the multidimensional setting. 
However, establishing the interactive fixed-effect rank in the presence of covariates is an open area of research. 

A useful heuristic in this setting is to perform tests from either \cite{BaiNg2002}, \cite{ahn2013eigenvalue}, \cite{babii2025tensor}, or from Section 5 of \cite{ke2024robust} over the residual term to gauge an approximation to multilinear rank. 
Conditions for these tests can be restrictive over the relative asymptotics of each $N_n$ with respect to $N_{n'}$ for $n'\neq n$, hence formalising theory for these tests is beyond the scope of this paper. 
Another useful heuristic is to analyse the scree plot of each matricised sub-problem. 
That is, perform \cite{Bai2009} for each $n=1,\dots, d$, and plot the distribution of singular values for the residual term. 
For finite multilinear rank tensors, this is likely to elucidate rank well, but formal results are again relegated to future research. 

Consider the factor model applied to a flattening that may not be low-rank. 
First, assume $\varphi^{(1)}$ varies in a high-dimensional parameter space, e.g. with $N_1 < N_2N_3$, $\varphi^{(1)} \in \mathbb{R}^{N_1\times N_1}$ and $\Gamma\in\mathbb{R}^{N_2N_3\times N_1}$ with each column mutually orthogonal for both these matrices. 
Then $\varphi^{(1)}\Gamma^\prime$ is full-rank and the factor model will not control for this. 
On the contrary, consider $\varphi^{(1)} \in \mathbb{R}^{N_1\times N_1}$ where all columns are linearly dependent. 
Then the matrix $\varphi^{(1)} \Gamma^\prime$ is rank-1 regardless of $L$ and of how $\varphi^{(2)}$ and $\varphi^{(3)}$ vary, thus can be projected with a factor model estimated with 1 factor. 
This situation is exemplified in simulations in Section~\ref{sect:sims}. 
The matrix method requires knowledge of this low rank dimension - the kernel weighted-within does not.

Knowing the multilinear rank,
\cite{MoonWeidner2015} show a factor model with at least $r_n$ factors results in consistent $\beta$ estimates.\footnote{This only requires knowing an upper bound on multilinear rank, not the true multilinear rank. } 
However, with generic heteroskedasticity and weak correlations, the factor model can only be bounded by a slow rate of convergence, which is exemplified in simulations below. 
In the three dimensional setting with proportional dimensions, this bound is as slow as $N^{-1/6}$, where $N$ is total sample size. 
This is too slow to use for standard inference procedures.\footnote{Debias do exist, but current proposals would achieve at best $N^{-1/3}$ convergence. There may be split sample debias innovations that can achieve sufficiently small asymptotic bias, but this is not explored. }

\subsection{Estimating kernel weight proxies}\label{sect:DiscussionProxyEst}

Discussed here are some functionals of multidimensional arrays that are useful for estimating proxies to form kernel weights. 
For this purpose, it is important to find proxies that isolate variation in each dimension since weights are formed separately for each index.

Consider for any of the dimensions $n$ the corresponding matrix of left singular vectors from above, $U^{(n)}$, estimated with noise $\varepsilon_{ijt}$. 
That is, each $U^{(n)}$ are calculated from the matrix $\boldsymbol{\mathcal{V}}_{(n)} = \boldsymbol{\mathcal{A}}_{(n)} + \boldsymbol{\varepsilon}_{(n)}$. 
Under regularity conditions on the noise term $\boldsymbol{\varepsilon}$, the left singular vectors from this decomposition consistently estimate $U^{(n)}$, up to rotations. 
In the three dimensional case,
define the $r_n$-vector $\widehat{U}^{(n)}_{i}$ as the $i$-{th} row of the left singular matrix of $\boldsymbol{\mathcal{A}} + \boldsymbol{\varepsilon}$ flattened in the $n$\textsuperscript{th} dimension. 
\cite{BaiNg2002} show the following ``up-to-rotation'' consistency result:
\begin{lemma}[Theorem 1 from \cite{BaiNg2002}]\label{lemma:cons_proxy}
For any fixed integer $k\geq 1$, there exists an $(r_n\times k)$ matrix $H^k_n$ with ${\rm rank}(H^k_n) = \min\{k,r_n\}$ and $C_n = \min\big\{\sqrt{N_n}, \prod_{n^\prime\neq n}\sqrt{N_{n^\prime}}\big\}$ such that for each $n$ under some regularity conditions 
\begin{align*}
    C_n^2\norm{\widehat{U}^{(n)}_{i_n} - H^{k\,\prime}_n\varphi^{(n)}_{i_n}}^2 = O_p(1).
\end{align*}
\end{lemma}

Hence, if the true error, $\boldsymbol{\mathcal{V}} = \boldsymbol{\mathcal{A}} + \boldsymbol{\varepsilon}$, is observed, this establishes consistency of $\widehat{U}^{(n)}_{i_n} $ as proxies. 
Matrices, $H^k_n$, in Lemma~\ref{lemma:cons_proxy} can be ignored as these do not impact convergence rates, see end of Section~\ref{sect:AppendixIterative}. 
However, the true error, $\boldsymbol{\mathcal{V}} = \boldsymbol{\mathcal{A}} + \boldsymbol{\varepsilon}$, is not observed. 
Instead, the observed error
$\widehat{\mathcal{V}}_{ijt} = Y_{ijt} - X_{ijt}\widehat{\beta} =  X_{ijt}(\beta^0 - \widehat{\beta}) + \mathcal{A}_{ijt} + \varepsilon_{ijt}$,
depends on the estimate $\widehat{\beta}$, hence should be accounted for. 
Proposition A.1 in \cite{Bai2009} provides, 
\begin{align*}
    \frac{1}{N_n}\sum_{i_n = 1}^{N_n}\norm{\widehat{U}^{(n)}_{i_n} - H^{k\,\prime}_n\varphi^{(n)}_{i_n}}^2 = O_p(\|\beta^0 - \widehat{\beta}\|^2) + O_p\left( \frac{1}{\min\{N_n, \prod_{n^\prime \neq n} N_{n^\prime}\}} \right),
\end{align*}
if $\widehat{U}^{(n)}_{i_n}$ come from least squares in \cite{Bai2009}. 
Ignoring $H^{k}_n$, 
$\frac{1}{N_n}\sum_{i_n = 1}^{N_n}\norm{\widehat{U}^{(n)}_{i_n} - \varphi^{(n)}_{i_n}}^2 = O_p\left( {N_n}^{-1}\right)$ 
under the rate restriction $N_n \lesssim \prod_{n^\prime \neq n} N_{n^\prime}$ from Section~\ref{sect:motivation}, using the convergence rate for $\|\beta^0 - \widehat{\beta}\|$ in Proposition~\ref{prop:bai}.
Then, $C_n^{-2}$ from Proposition~\ref{lemma:consKernel} is $1/N_n$ and the parametric rate for the kernel weighted estimator can be established.  

\subsection{Bandwidth Selection}\label{sect:bandwidth}

This paper does not study formal theory behind any bandwidth optimisation tool, however, some practical guidance is useful. 
The estimator solves a semiparametric problem, hence a wide range of bandwidths still admit asymptotically negligible bias for slope coefficients. 
For example, the dimension specific rate of convergence is $O_p(h_n + N_n^{-1}h_n^{-1/2})$ when also considering variance. 
For $N_nh_n \rightarrow \infty$, the term reduces to $O_p(h_n) + o_p(N_n^{-1/2})$, the rates analysed in earlier sections. 
Then any rate with $N_nh_n \rightarrow \infty$ and $N_n^{1/2}h_n \rightarrow 0$ admits estimation error $o_p(N_n^{-1/2})$. Hence, asymptotically, the estimator is robust to many different specifications of bandwidths. 

In finite samples it is useful to have a data-driven rule. 
For this purpose, consider leave-one-out cross validation (LOOCV) of the weighted-within smoother applied to $\boldsymbol{Y}$, see e.g. \cite{hastie2008kernel}. 
The transformation $\check{\boldsymbol{Y}} = \boldsymbol{Y}\times_1 M_1 \times_2\dots\times_d M_d$ is linear in $\boldsymbol{Y}$, with residual operator the Kronecker product of the $M_n$, so the usual leave-one-out shortcut for linear smoothers gives the criterion
\begin{align*}
    \frac{1}{N}\sum_{i_1}\dots\sum_{i_d} \left(\frac{\check{Y}_{i_1\dots i_d}}{\prod_{n=1}^d M_{n,i_ni_n}}\right)^2,
\end{align*}
minimised over a common bandwidth, $h_{CV}$. 
When the weights are built one proxy component at a time, as in Appendix~\ref{sect:AppendixIterative}, each $M_n$ is the product of the matrices applied along dimension $n$ in a single pass. 
This criterion targets the mean squared error of $\widehat{\boldsymbol{\Gamma}}_{Y}$, whereas the admissible rates above are set by the bias in $\widehat{\beta}$, which permits undersmoothing. 
As in the simulations of \cite{freeman2026inference}, the bandwidth is therefore undersmoothed to $h = h_{CV}(\min_n N_n)^{-2/5}$, which lies inside the admissible range above whenever $h_{CV}$ shrinks more slowly than $N_n^{-3/5}$ and faster than $N_n^{-1/10}$.

\subsection{Projection Using the Linear Kernel}

As a special case, consider the weighted-within transformation with a linear kernel, $k: \mathcal{U}\times \mathcal{U} \to \mathbb{R}$ using estimated singular vectors $\widehat{U}_{i_n}^{(n)}$: $k(u_i, u_j) = u_i'u_j$.
Then in the transformation, 
\begin{align*}
    \check{\mathbf{Y}}
    &:=  
    \mathbf{Y}\times_1 M_1 \times_2 M_2\times_3\dots\times_d M_d, 
    =
    \mathbf{Y}\bigtimes_{n =1}^d (\mathbb{I}_{N_n} - W_n),
\end{align*}
each $W_n$ is replaced by $\widehat{U}^{(n)}\widehat{U}^{(n)\,'}$. 
By construction, these singular vectors are orthonormal, hence, this is equivalent to $W_n = \widehat{U}^{(n)}\left(\widehat{U}^{(n)\,'}\widehat{U}^{(n)} \right)^{-1} \widehat{U}^{(n)\,'}$. 
Then, 
\begin{align*}
    \mathbf{Y}\bigtimes_{n =1}^d (\mathbb{I}_{N_n} - W_n)
    = 
    \mathbf{Y}\bigtimes_{n =1}^d 
    \left(\mathbb{I}_{N_n} - \widehat{U}^{(n)}\left(\widehat{U}^{(n)\,'}\widehat{U}^{(n)} \right)^{-1} \widehat{U}^{(n)\,'}\right)
\end{align*}
is simply a multilinear projection of $\{\widehat{U}^{(n)}\}_{n=1}^d$. 

This special case of the weighted-within estimator has two main advantages, but one major disadvantage. 
First, the error can be very simply bounded in this case by arguments presented in Appendix~\ref{sect:LinearKernel}. 
Also, this kernel comes with no bandwidth parameters to tune, hence is much simpler to implement. 
However, the linear kernel is not flexible, and may miss generic smooth variation for how fixed-effects enter the data generating process for $\boldsymbol{X}$. 
For this reason, this paper focuses on the more flexible version using generic kernel functions.
The flexible version of the estimator likely also has applications extending beyond the parametric model for interactive effects in $\boldsymbol{Y}$, and certainly requires fewer restrictions on the data generating process for $\boldsymbol{X}$ in this paper. 
The linear kernel case is included for completeness in Appendix~\ref{sect:LinearKernel}, but a more complete analysis including implementation is left for future work.

\section{Empirical application - demand estimation for beer}\label{sect:Application}

The proposed methods are applied to estimate the demand elasticity for beer. 
Price and quantity for beer sales is taken from the Dominick's supermarket dataset for the years 1991-1995 and is related to supermarkets across the Chicago area. 
Price and quantity vary over three dimensions in this example -- product ($i$), store ($j$) and fortnights ($t$). 
Fixed-effects that interact across all three dimensions can control for temporal taste shocks to beer consumption that differ over both product and store. 
Take for instance a large sporting event (temporary $t$ shock) that changes preferences differently across locations ($j$) and across certain subsets of sponsored beer ($i$). 
Analysts may want some way to control for such heterogeneity that interacts over all dimensions, even if they do not explicitly observe such sponsorship variables. 
For example, 
in the stadiums for the many NBA finals playoffs the Chicago Bulls played in the early 1990's, 
Miller Lite beer advertisements could be seen alongside advertisements for a substitute product, Canadian Club whisky. 
This suggests these events attracted large marketing campaign spends for these and other beer substitute brands that most likely also included price offers at local supermarkets. 
Whilst the impact of these advertisements and price offers on the demand for or price of beer is not clear and,
further, that it is reasonably safe to assume the econometrician does not observe the plethora of marketing campaigns around these events,
the analyst would most likely still want to control for aggregate shocks like these. 

In their critique of the fixed-effects approach to demand estimation, \cite{berry2021foundations} note the approach in its simplistic additive form cannot control effects that interact over products $\times$ markets, see Section 2.5.1 in \cite{berry2021foundations}. 
This limitation is significantly relaxed with interactive terms over all three dimensions. 
Indeed, Table~\ref{tbl:demandappLogLog} suggests even products $\times$ markets interactions are limiting - considering products $\times$ markets $\times$ time substantially shifts point estimates, cf. ``Additive Fixed-effects'' versus ``Weighted-within'' estimates. 

Price and quantity are jointly determined, so the identifying assumption is that the demand shocks that also move prices are low-rank over $ijt$, and hence absorbed by $\boldsymbol{\mathcal{A}}$.
The interactive fixed-effects identification strategy is not perfect - indeed there may be endogenous dynamics that truly vary with full-rank over all $ijt$.
However, valid, and strong, instruments can be scarce.
Interactive fixed-effects provide a flexible model to control these effects if their variation is in any way systematic over $i$, $j$, or $t$. 
\cite{moon2018estimation} find similarly in the BLP model for demand on automobiles that conditional on even one factor component, estimates do not change significantly whether or not an instrument is used, i.e. if prices are still considered endogenous conditional on interactive effects. 
Hence, estimates from these models are at least a useful comparison to instrumental variable methods. 

Models for demand estimation ideally account for endogenous variation in prices and quantity. 
The classic instrumental variable approach is to find a supply shifter that shifts the supply curve, allowing the econometrician to trace out the slope of the demand curve. 
A popular instrument in the estimation of beer demand is the commodity price for barley, one of the product's main ingredients, see e.g. \citet{saleh2014simple,tremblay1995advertising, richards2021dynamic}. 
Since the price of barley is arguably not driven by the demand for it by any one supplier of beer, it can be a useful variable to instrument for price shifts. 
In the following, it is taken as given that the price of barley is exogenous with respect to the fixed-effects and noise term, $\varepsilon$.
The pooled IV estimator regresses the monthly average of log quantity over products and stores on the monthly average of log price, instrumented by the log price of barley, with an intercept.
It includes no fixed-effects, since the instrument varies only over time, so the first stage fitted price varies only over time and cannot identify the slope alongside a time effect.
Its standard error in Table~\ref{tbl:demandappLogLog} is heteroskedasticity robust, since the HAC standard error, built for the three-dimensional panel, does not apply to it.

For valid inference, the instrument should be highly correlated with the price of beer. 
In this dataset, correlation between the price of barley, which varies over only month, every second $t$, and price of beer depends on how beer price is first aggregated. 
If beer price is first integrated over $i$, $j$, and to the monthly level, such that it only varies over every second $t$, then it is highly correlated with the price of barley, at 0.73. 
However, if beer price is not aggregated at all it is only correlated at 0.06. 
This heuristic suggests there are important product and store level price drivers for beer that are not accounted for by fluctuations in the price of barley, or indeed by price fluctuations only over time. 
Standard errors for the IV estimator below are much larger than other estimators, which may be explained partly by this loss in effective sample size.\footnote{The effective sample size for the IV estimator drops from $N_1N_2T$ to just 57 months. }

Table~\ref{tbl:demandappLogLog} refers to estimates from the following model, 
\begin{align*}
    \log(quantity_{ijt}) = \log(price_{ijt})\beta + \sum_{\ell = 1}^L \lambda_{i\ell}  \delta_{j\ell}  \gamma_{t\ell} + \varepsilon_{ijt}.
\end{align*}
No additional controls are included here since they are low-dimensional and subsumed by additive fixed-effects. 
About 5\% of product-store-fortnight cells have zero or missing sales. Their price is set to the sample mean and their quantity to half the smallest positive quantity, so that the fixed-effects transformations act on a complete array, and every slope in Table~\ref{tbl:demandappLogLog} is estimated on the observed cells only. 
Estimates for pooled OLS are positive, a contradiction that demand curves are downward sloping.
This reflects unobserved demand shifters, such as product quality and local demand shocks, that raise price and quantity together, and the estimate is precise only because the sample is large.
IV estimates are negative, but standard errors are large. 
As noted above, the effective sample size for the second stage of the instrumental variable approach is orders of magnitude smaller because the instrument only varies over time, specifically every second fortnight. 
Estimates under the additive model for fixed-effects are close to zero, with smaller standard errors than the IV estimate. 

\begin{table}[h]
\centering
    \setlength{\tabcolsep}{3.5pt}
    \begin{tabular}{|l|ccc|}
  \hline
  Estimator &  &  & 
  \\
  \hline 
    Pooled OLS & \multicolumn{3}{c|}{1.304 (0.010)} \\ 
  Additive Fixed-effects & \multicolumn{3}{c|}{-0.069 (0.500)} \\ 
  Pooled IV & \multicolumn{3}{c|}{-3.393 (1.434)} \\ 
  \hline
  Rank & $r = 3$ & $r = 4$ & $r = 5$ 
  \\ 
  \hline
  \textbf{Weighted-within} (Iterative, this paper) & -2.990 (0.220) & -3.414 (0.301) & -4.114 (0.366) \\
  Factor (Product (i) as rows) & -2.167 (0.337) & -2.682 (0.339) & -2.786 (0.336) \\ 
  Factor (Store (j) as rows) & -0.122 (0.366) & -0.269 (0.305) & -0.287 (0.314) \\ 
  Factor (Time (t) as rows) & -0.180 (0.294) & -0.215 (0.289) & -0.236 (0.288) \\ 
   \hline
\end{tabular}
\caption{
Log-log demand elasticities (45 products, 48 stores, 110 fortnights), with $r$ interactive terms for the interactive fixed-effects estimators.\\
{\footnotesize
HAC standard errors in brackets, except for pooled IV, which uses monthly averages over the monthly sample (53 products, 48 stores, 57 months) and reports a heteroskedasticity robust standard error. 
}
}\label{tbl:demandappLogLog}
\end{table}

\begin{figure}[ht]
    \centering
    {Factor models vs. Weighted-within over number of factors}\\
    \includegraphics[width = 0.8\linewidth]{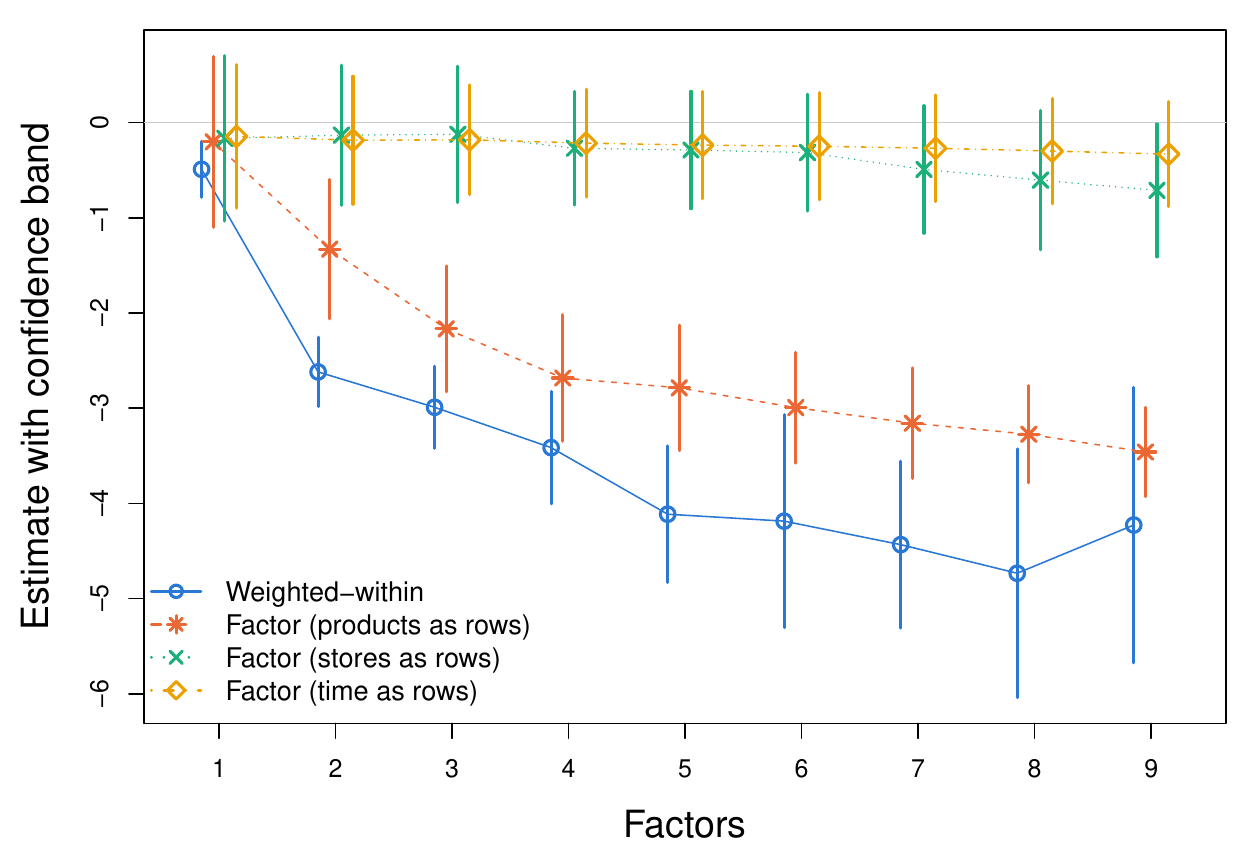}
    \caption{Elasticities over number of factors with HAC 95\% confidence intervals.}
    \label{fig:factorvsWW}
\end{figure}

The matrix method is implemented in all three dimensions, 
i.e., transform the three-dimensional array into three different matrices - with products as rows, with stores as rows, and time as rows. 
Three to five factors are estimated for each. 
The results suggest the factor model for fixed-effects is very sensitive to how the array is organised into a two-dimension problem, with vastly different conclusions drawn from each specification. 
Simulations in Section~\ref{sect:sims} suggest an explanation - sparsity of the fixed-effects in the interaction term can be heterogeneous, and imply a different factor model rank when the tensor of data is flattened, or matricised, in different dimensions.

This paper's novel estimator, the weighted-within, estimates elasticities similar to the IV point estimates, but with much smaller standard errors.\footnote{
The bandwidth is $h = 0.34$, from the rule in Section~\ref{sect:bandwidth}, with $h_{CV} = 1.54$ undersmoothed by $45^{-2/5}$, and is held fixed across $r$. Each scalar proxy is standardised to unit variance before the kernel is applied, so $h$ is in standard-deviation units.
} 
Weighted within estimates are more elastic than the own-price elasticities from Table 1 in \cite{hausman1994competitive}, which average around $-2.5$.

Figure~\ref{fig:factorvsWW} displays estimates and confidence intervals for the weighted-within estimator and the three factor models, allowing the number of interactive terms to increase up to nine. 
Beyond nine terms the weighted-within standard errors become unstable. 
The figure displays two things. 
First, the weighted-within estimates move little beyond about five interactive terms, while the factor model estimates keep moving as terms are added, with those for stores or time as rows much closer to zero. 
Second, the weighted-within estimates are more negative than every factor model at each number of interactive terms. 
This implies a persistent debias with the weighted-within transformation. 
From five terms, the weighted-within confidence intervals are wider than those of the factor model with products as rows.

\section{Simulations}\label{sect:sims}

Two simulation exercises are considered. 
The first analyses performance 
when sample size is allowed to increase. 
The second analyses fixed sample properties of the estimators when Assumption~\ref{ass:multiLrank}, 
which governs fixed-effect multilinear rank,
is violated in some dimensions. 

\subsection{Growing sample exercise}

Data is generated as, 
\begin{align}\label{eqn:simdgp1}
    Y_{ijt} &= X_{ijt} + \mathcal{A}_{ijt} + \varepsilon_{ijt}\\
    \nonumber
    X_{ijt} &=  
    2\mathcal{A}_{ijt}
    -\rho
    \sum_{\ell = 1}^{2} 
    \Big(\lambda_{i\ell} + \lambda_{i-1\ell}\Big)
    \Big(\gamma_{j\ell} + \gamma_{j-1\ell}\Big)
    \Big(f_{t\ell} + f_{t-1,\ell} \Big)
    + \eta_{ijt}
    \\
    \nonumber
    \varepsilon_{ijt} &= \Big(1/\sqrt{2}\Big)\big(\nu_{ijt} + \nu_{ijt-1} +\nu_{ij-1t} +\nu_{ij-1t-1} +\nu_{i-1jt} +\nu_{i-1jt-1} +\nu_{i-1j-1t} +\nu_{i-1j-1t-1}\big)
\end{align}
with, 
$\mathcal{A}_{ijt} = \sum_{\ell = 1}^{2} \lambda_{i\ell}  \gamma_{j\ell}  f_{t\ell} $. Also, 
\begin{align*}
    \eta_{ijt} \stackrel{ i.i.d.}{\sim} N(0,1), \nu_{ijt} \stackrel{ i.ni.d.}{\sim} N(0,\eta_{ijt}^2) & \textrm{ and for each } \ell, \,\, \lambda_{i\ell}, \gamma_{j\ell}, f_{t\ell} \stackrel{ i.i.d.}{\sim} \,\,N(0,1). 
\end{align*}
Error $\varepsilon_{ijt}$ admits heteroskedasticity, and correlation in each dimension of the data. 
The order of $i$ and $j$, are randomised to simulate an unknown cross correlation structure. 

The results depicted in Figure~\ref{fig:asymptotic} show two specifications for $\rho$. 
The left panel depicts $\rho = 1$, and the right panel depicts $\rho = 1/3$. 
Estimators considered use 2 estimated factors such that the multilinear rank is correctly predicted. 
The left panel shows that whilst the matrix methods, labelled Factor, look consistent, they converge at a rate much slower than the empirical bounds. 
The right panel shows more egregious convergence when the bias problem is made worse with $\rho = 1/3$. 
Indeed, for the sample sizes considered, the empirical bounds do not ever cover the true parameter value in either panel. 

This contrasts with the results for the kernel weighted fixed-effect projection. 
Finite sample bias is negligible compared to variance in both the left and right panel. 
This estimator has roughly the same variance as the matrix methods, and for the sample sizes considered, the empirical bounds are always valid. 

\begin{figure}[h]
    \begin{center}
    \includegraphics[scale = 0.4]{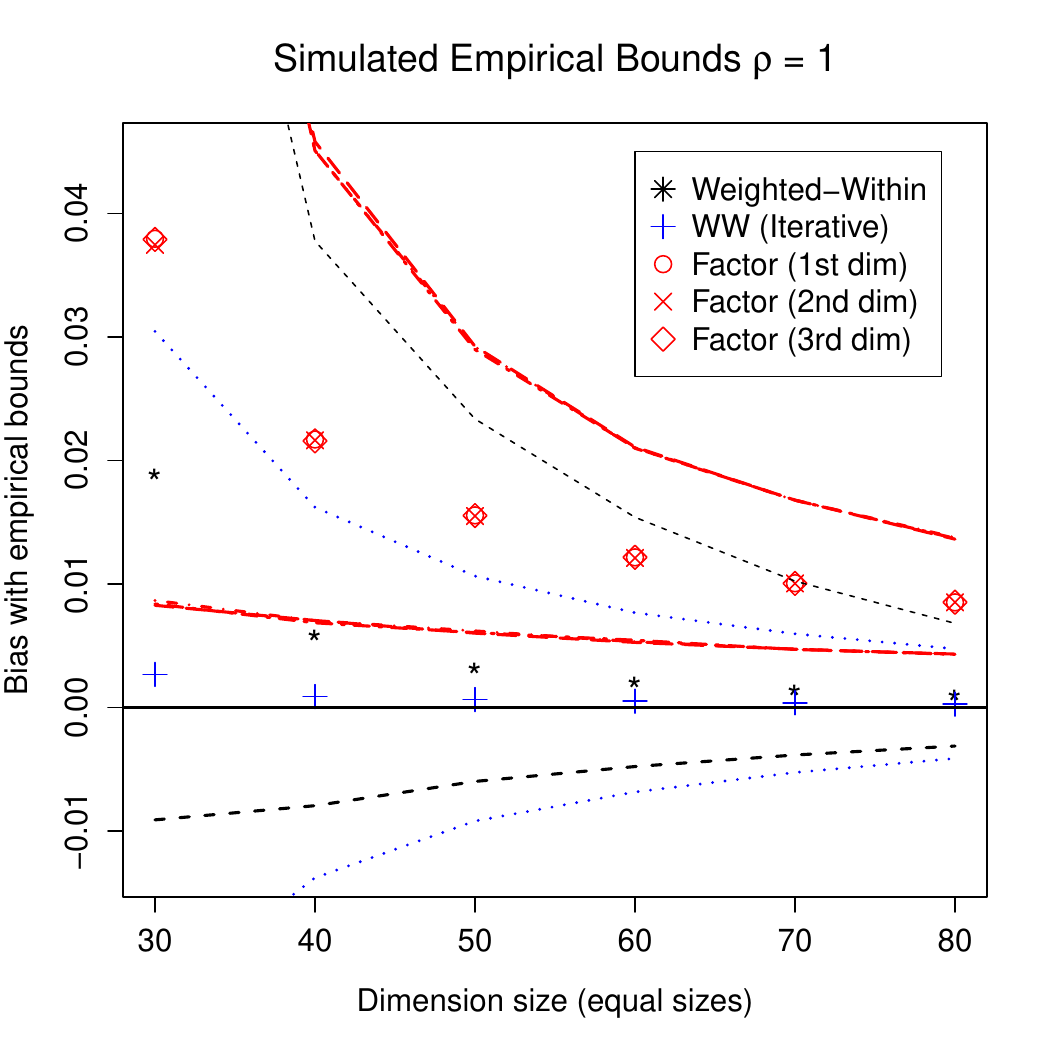}
    \includegraphics[scale = 0.4]{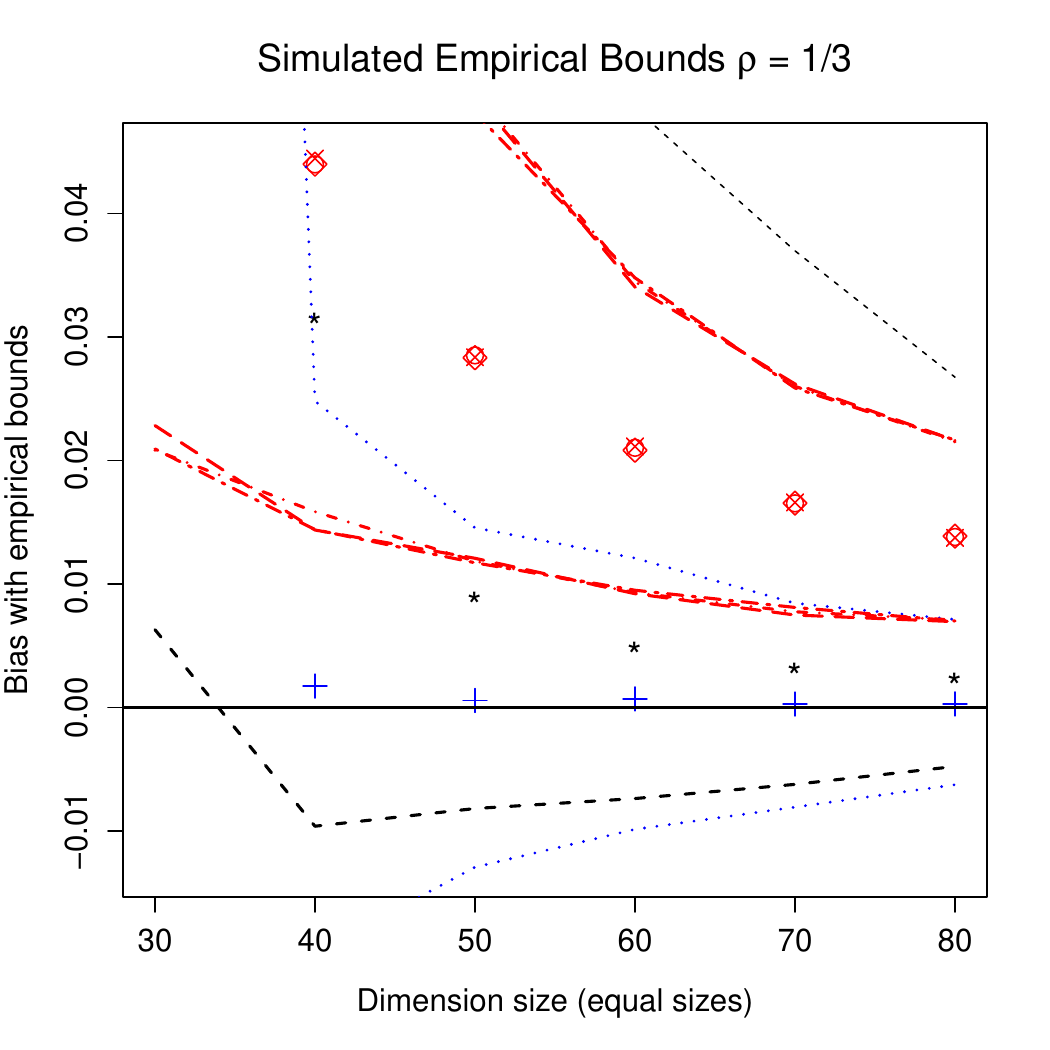}
    \end{center}
    \vspace{-0.5cm}
    \caption{Bias of $\beta$ estimates with 95\% empirical bounds for DGP in \eqref{eqn:simdgp1}}\label{fig:asymptotic}
    \footnotesize{Left panel, $\rho = 1$ 10,000 Monte Carlo rounds, right panel $\rho = 1/3$ 1,000 Monte Carlo rounds. 
    Points depict mean bias. 
    Dotted lines depict 95\% empirical bounds. 
    Factor \# depicts the matrix factor methods with dimension \# as rows. 
    Weighted-Within and WW (Iterative) depict the kernel weighted method and its iterative version from Appendix~\ref{sect:AppendixIterative}.
    }
\end{figure}

Table~\ref{tbl:simCoverHom} displays coverage and compares the weighted estimator, its iterative version from Appendix~\ref{sect:AppendixIterative}, and the factor model applied to all dimensions.
The W (Iter) coverage, which is coverage for the iterative estimator, is always better than the W coverage, and achieves nominal coverage for larger sample sizes when the multilinear rank is exactly or over estimated. 
There is evidence of slight over-coverage when the rank is overestimated, but this is expected when undersmoothing. 
Coverage for W improves with overestimation of the rank. 
The factor model undercovers with progressively worse coverage as the sample size grows. 

\begin{table}
\centering
\begin{tabular}{c|ccccc||ccccc|}
    \hline
  &  $\widehat{r} = 2$  & \multicolumn{4}{c}{Coverage} & $\widehat{r} = 3$  & \multicolumn{4}{c}{Coverage}\\  
\hline
Dim & W & W (Iter) & F1  & F2 & F3 & W & W (Iter) & F1  & F2 & F3 \\ 
  \hline
  10 & 0.03 & 0.45 & 0.11 & 0.10 & 0.14 & 0.02 & 0.68 & 0.14 & 0.15 & 0.23 \\ 
  20 & 0.12 & 0.76 & 0.07 & 0.06 & 0.12 & 0.19 & 0.83 & 0.07 & 0.07 & 0.17 \\ 
  30 & 0.42 & 0.95 & 0.04 & 0.04 & 0.11 & 0.64 & 0.96 & 0.04 & 0.05 & 0.12 \\  
  40 & 0.56 & 0.95 & 0.03 & 0.02 & 0.09 & 0.78 & 0.95 & 0.03 & 0.02 & 0.09 \\ 
  50 & 0.61 & 0.96 & 0.01 & 0.01 & 0.06 & 0.84 & 0.96 & 0.01 & 0.01 & 0.07 \\
  60 & 0.63 & 0.96 & 0.01 & 0.01 & 0.06 & 0.84 & 0.95 & 0.01 & 0.00 & 0.06 \\  
  70 & 0.64 & 0.95 & 0.00 & 0.00 & 0.04 & 0.90 & 0.96 & 0.00 & 0.00 & 0.04 \\  
  80 & 0.66 & 0.96 & 0.00 & 0.00 & 0.03 & 0.88 & 0.95 & 0.00 & 0.00 & 0.04 \\  
   \hline
\end{tabular}
\caption{Estimated coverage of 95\% confidence intervals with true rank and overestimated rank \\
\footnotesize{
True multilinear rank is (2,2,2) and $\rho = 1$. Left panel 10,000 Monte Carlo rounds, right panel 500 Monte Carlo rounds. 
Left panel estimates true rank, right panel estimates rank + 1. 
Dim is the size of each dimension, with equal dimension sample sizes. 
W is the weighted-within estimator of Section~\ref{sect:estimatorsKernel}. W (Iter) is the iterative weighted estimator. F\# is the matrix method making dimension \# the rows, and the other dimensions jointly the columns. Coverage is for HAC standard errors.  
}
}\label{tbl:simCoverHom}
\end{table}

\subsection{Fixed sample exercise}

Table~\ref{tbl:sim} shows simulation results for the following DGP, 
\begin{align*}
    \begin{split} 
        Y_{ijt} &= X_{ijt}\beta + \mathcal{A}_{ijt} + \varepsilon_{ijt}\\
        X_{ijt} &=  \mathcal{A}_{ijt} + 
        \sum_{\ell = 1}^{N_1} 
        \Big(\lambda_{i\ell} + \lambda_{i-1\ell}\Big)
        \Big(\gamma_{j\ell} + \gamma_{j-1\ell}\Big)
        \Big(f_{t\ell} + f_{t-1,\ell} \Big) + \eta_{ijt}
    \end{split}
\end{align*}
with $\mathcal{A}_{ijt} = \sum_{\ell = 1}^{N_1} \lambda_{i\ell}  \gamma_{j\ell}  f_{t\ell} $, $\eta_{ijt}$ as in \eqref{eqn:simdgp1}, and $\varepsilon_{ijt}$ heteroskedastic and correlated in each dimension. 
$\mathcal{A}_{ijt}$ is
normalised to have unit variance.
$\boldsymbol{\mathcal{A}}$ is specified such that it is rank 1 when flattened in the first dimension and rank $N_1$ when flattened in either dimension two or three. 
That is, the multilinear rank is $\mathbf{r} = (1,N_1,N_1)$. 
To achieve this, the matrix $\lambda$ is designed to be rank-1 and the matrices $\gamma$ and $f$ are designed to be rank-$N_1$. 

In Table~\ref{tbl:sim}, the estimators OLS and Fixed-effects are simply the pooled OLS estimator and the pooled OLS estimator after additive fixed-effects are projected out, respectively. 
As expected both of these have bias. 
The factor model is used after first flattening along each dimension as Factor(dim = $n$), where $n$ is the dimension used for flattening. 
In each case, one factor is projected. 
The results show the bias is close to zero when the correct dimension is flattened over (the first dimension in this case) and very poor bias when the incorrect dimension is used (the second and third dimensions). 
Lastly, the weighted differencing estimator is estimated with Gaussian kernel function with various bandwidths; which are standardised to be equivalent to standard deviations of the proxy measures. 
Weighted estimators show a clear bias-variance trade-off, but all with bias of the same order as the correct factor model. 

This analysis is repeated for the four dimensional case in Table~\ref{tbl:sim}, where the first and second dimensions admit low-dimensional unobserved interactive fixed-effects parameters. 
The simulations suggest similar results as the three dimensional case, where the factor models perform well when flattened in the low-dimensional dimensions (first and second) and poorly in the high-dimensional dimensions (third and fourth). 

\begin{table}
\centering
\begin{tabular}{|l|ccc||c|ccc|}
  \hline
3-D  & Bias & St. dev. & RMSE & 4-D  & Bias & St. dev. & RMSE  \\ 
  \hline
    OLS & 0.3655 & 0.0038 & 0.3655 && 0.1197 & 0.0212 & 0.1216  \\ 
  Fixed-effects & 0.3709 & 0.0040 & 0.3709 && 0.1335 & 0.0264 & 0.1361 \\ 
  WW (h = 0.1) & 0.0008 & 0.0066 & 0.0067 && 0.0001 & 0.0042 & 0.0042  \\ 
  WW (h = 0.25) & 0.0046 & 0.0051 & 0.0069 && 0.0002 & 0.0025 & 0.0025  \\ 
  WW (h = 0.5) & 0.0298 & 0.0063 & 0.0305 && 0.0050 & 0.0024 & 0.0055 \\ 
  WW (h = 1) & 0.1538 & 0.0133 & 0.1544 &&  0.0525 & 0.0087 & 0.0532 \\ 
  WW(Iter.) (h = 0.1) & -0.0025 & 0.0182 &0.0184 && 0.0011 & 0.0389 & 0.0389 \\ 
  WW(Iter.) (h = 0.25) & -0.0023 & 0.0076 & 0.0079 && 2.999e-05 & 0.0054 & 0.0054 \\ 
  WW(Iter.) (h = 0.5) & -0.0008 & 0.0056 & 0.0056 && -1.946e-05 & 0.0030 & 0.0030 \\ 
  WW(Iter.) (h = 1) & 0.0013 & 0.0048 & 0.0050 && 2.315e-05 & 0.0022 & 0.0022  \\ 
  Factor1 & -0.0028 & 0.0041 & 0.0049 && -0.0014 & 0.0016 & 0.0021 \\ 
  Factor2 & 0.3604 & 0.0064 & 0.3605 && -0.0014 & 0.0016 & 0.0022  \\ 
  Factor3 & 0.3605 & 0.0064 & 0.3605 && 0.1219 & 0.0268 & 0.1248 \\ 
  Factor4 & NA & NA & NA && 0.1219 & 0.0271 & 0.1249  \\
   \hline
\end{tabular}
\caption{Fixed-sample size simulation with heterogeneous multilinear rank\\
\footnotesize{3D model ($N_1 = N_2 = N_3 = 40$), with 10,000 Monte Carlo rounds. \\
4D model ($N_1 = N_2 = N_3 = N_4 = 20$), with 1,000 Monte Carlo rounds. Bandwidth scaled up by $\sqrt{2}$ to account for smaller per dimension sample.
Factor4 is not defined in the 3D model.
All results relate to $\beta$ estimation. }
}\label{tbl:sim}
\end{table}

\section{Conclusion}\label{sect:conclusion}

This paper develops methods to generalise the interactive fixed-effect to multidimensional datasets with more than two dimensions. 
Theoretical results show that standard matrix methods can be applied to this setting but require additional knowledge of the data generating process and generally have slow convergence rates. Nonetheless, these provide useful preliminary estimates. 
The multiplicative interactive fixed-effect error from the kernel weighted method show an improvement on the convergence rate of slope coefficient estimates to the parametric rate, and suggest a more robust approach to projecting fixed-effects. 
Simulations show finite sample properties when the sample size is allowed to grow and when it is fixed. 
These simulations exemplify how sensitive the two-dimensional methods are to model specification issues as simple as how to organise the dataset. 
They also show the robustness of the kernel weighted fixed-effects estimators without having to make these same specifications. 
A method for inference with the weighted fixed-effects estimator is also introduced. 

The model is applied to a demand model for beer consumption. 
The application demonstrates that simply applying the two-dimensional factor model approach is sensitive to how dimensions are arranged to suit these estimators. 
The weighted-within transformation estimates elasticities close to an instrumental variable point estimate, but with substantially better precision.

\section*{Acknowledgements}

I would like to greatly thank Martin Weidner. 
I would also like to thank Lars Nesheim. 
Additionally, I would like to thank Tim Christensen, Ivan Fern\'andez-Val, Dennis Kristensen, Matthew Read, and Andrei Zeleneev for their helpful suggestions along with the audience at IPDC 2022, IAAE Conference 2022, the Bristol Econometrics Study Group 2022, the MEG 2022/2023, and the NASMES 2024. 
This research was supported by the European Research Council grant ERC-2018-CoG-819086-PANEDA.

\appendix

\section{Appendix: Proofs}\label{sect:consGFEproof}

\renewcommand{\thetheorem}{A.\arabic{theorem}} 
\setcounter{theorem}{0}
\renewcommand{\thelemma}{A.\arabic{lemma}} 
\setcounter{lemma}{0}
\renewcommand{\theproposition}{A.\arabic{proposition}} 
\setcounter{proposition}{0}
\renewcommand{\theequation}{A.\arabic{equation}} 
\setcounter{equation}{0}
\renewcommand{\theassumption}{A.\arabic{assumption}} 
\setcounter{assumption}{0}

\subsection{Proofs of main paper results}

\begin{proof}[\bf Proof of Proposition~\ref{lemma:consKernel}]

Let ${\rm vec}_K\big(\widetilde{\mathbf{X}}\big)$ be the $\prod_nN_n\times K$ matrix of vectorised covariates after weighted-within. 
    Operator ${\rm vec}$ is standard vectorisation.
    Let $N=\prod_nN_n$ and the subscript $i = 1,\dots, N$ be the index for the vectorised data when $i$ has no subscript. 
    Then,
    \begin{align*}
        \norm{\beta_{\mathcal{W}} - \beta^0} &= \Big\| \left({\rm vec}_K\big(\widetilde{\mathbf{X}}\big)^\prime {\rm vec}_K\big(\widetilde{\mathbf{X}}\big) \right)^{(-1)} {\rm vec}_K\big(\widetilde{\mathbf{X}}\big)^\prime \left({\rm vec}\big(\boldsymbol{\widetilde{\mathcal{A}}}\big) + {\rm vec}\big(\widetilde{\boldsymbol{\varepsilon}}\big) \right) \Big\| 
        \leq \norm{\kappa_N} + \norm{\omega_N}
    \end{align*}
    where 
    \begin{align*}
        \norm{\kappa_N} &:=  \Big\| \left({\rm vec}_K\big(\widetilde{\mathbf{X}}\big)^\prime {\rm vec}_K\big(\widetilde{\mathbf{X}}\big) \right)^{(-1)} {\rm vec}_K\big(\widetilde{\mathbf{X}}\big)^\prime {\rm vec}\big(\boldsymbol{\widetilde{\mathcal{A}}}\big)  \Big\|; \\
        \norm{\omega_N} &:=  \Big\| \left({\rm vec}_K\big(\widetilde{\mathbf{X}}\big)^\prime {\rm vec}_K\big(\widetilde{\mathbf{X}}\big) \right)^{(-1)}{\rm vec}_K\big(\widetilde{\mathbf{X}}\big)^\prime {\rm vec}\big(\widetilde{\boldsymbol{\varepsilon}}\big) \Big\|. 
    \end{align*}
    $\norm{\omega_N}$ is bounded from Assumption~\ref{ass:regCondKer}, and results in Appendix~\ref{sect:AppendixSupp}. Focus on $\norm{\kappa_N}$.
    By Assumption~\ref{ass:regCondKer}.(i), $\big(N^{-1}{\rm vec}_K(\widetilde{\mathbf{X}})^\prime{\rm vec}_K(\widetilde{\mathbf{X}})\big)^{-1} = O_p(1)$ and $N^{-1}\sum_{i=1}^N\widetilde{X}_{i,k}^2 = O_p(1)$ for each $k$, so by the Cauchy-Schwarz inequality,
    \begin{align*}
        \norm{\kappa_N}
        \leq
        O_p(1)\max_k\left|\frac{1}{N}\sum_{i=1}^N \widetilde{X}_{i,k}\widetilde{\mathcal{A}}_i\right|
        \leq
        O_p(1)\max_k\left(\frac{1}{N}\sum_{i=1}^N \widetilde{X}_{i,k}^2\right)^{1/2}\frac{1}{\sqrt{N}}\big\|\widetilde{\boldsymbol{\mathcal{A}}}\big\|_F
        =
        O_p(1)\frac{1}{\sqrt{N}}\big\|\widetilde{\boldsymbol{\mathcal{A}}}\big\|_F,
    \end{align*}
    where $\widetilde{\boldsymbol{\mathcal{A}}} = \boldsymbol{\mathcal{A}}\times_1 M_1\times_2\dots\times_d M_d$ with $M_n = \mathbb{I}_{N_n} - W_n$ as in Section~\ref{sect:estimatorsKernel}.

    For each $n\in\mathcal{M}^\prime$ let $\varphi^{(n)+} := (\varphi^{(n)\prime}\varphi^{(n)})^{-1}\varphi^{(n)\prime}$.
    The columns of $\mathcal{A}_{(n)} = \varphi^{(n)}\Gamma_n^\prime$ lie in the column space of $\varphi^{(n)}$, and $n^\prime$-mode products with $n^\prime\neq n$ preserve this.
    Hence $\boldsymbol{\mathcal{A}} = \boldsymbol{\mathcal{B}}\bigtimes_{n\in\mathcal{M}^\prime}\varphi^{(n)}$ with $\boldsymbol{\mathcal{B}} := \boldsymbol{\mathcal{A}}\bigtimes_{n\in\mathcal{M}^\prime}\varphi^{(n)+}$, since $n$-mode products in distinct dimensions commute.
    Then
    \begin{align*}
        \widetilde{\boldsymbol{\mathcal{A}}}
        =
        \boldsymbol{\mathcal{B}}\bigtimes_{n\in\mathcal{M}^\prime}\big(M_n\varphi^{(n)}\big)\bigtimes_{n\notin\mathcal{M}^\prime} M_n .
    \end{align*}
    For any tensor $\mathbf{T}$ and matrix $Q$, $(\mathbf{T}\times_n Q)_{(n)} = Q\mathbf{T}_{(n)}$, so $\|\mathbf{T}\times_n Q\|_F\leq\|Q\|\,\|\mathbf{T}\|_F$.
    Applying this in each dimension, and to $\boldsymbol{\mathcal{B}}$,
    \begin{align}\label{proof:tuckerBound}
        \frac{1}{\sqrt{N}}\big\|\widetilde{\boldsymbol{\mathcal{A}}}\big\|_F
        \leq
        \frac{1}{\sqrt{N}}\|\boldsymbol{\mathcal{A}}\|_F
        \prod_{n\in\mathcal{M}^\prime}\big\|\varphi^{(n)+}\big\|\,\big\|M_n\varphi^{(n)}\big\|
        \prod_{n\notin\mathcal{M}^\prime}\|M_n\|.
    \end{align}
    The first factor is $O_p(1)$ by assumption.
    For $n\notin\mathcal{M}^\prime$, the rows of $W_n$ sum to one, so $\|M_n\|\leq 1 + \|W_n\| \leq 1 + \big(\max_{j_n}\sum_{i_n}W^{(n)}_{i_n,j_n}\big)^{1/2} = O_p(1)$ by Assumption~\ref{ass:regCondKer}.(iv).
    For $n\in\mathcal{M}^\prime$, $\|\varphi^{(n)+}\| = O_p(N_n^{-1/2})$ since the smallest eigenvalue of $N_n^{-1}\varphi^{(n)\prime}\varphi^{(n)}$ is bounded away from zero wpa1.
    Row $i_n$ of $M_n\varphi^{(n)}$ is $\sum_{j_n}W^{(n)}_{i_n,j_n}\big(\varphi^{(n)}_{i_n} - \varphi^{(n)}_{j_n}\big)$, and the weights are nonnegative and sum to one, so by Jensen's inequality
    \begin{align*}
        \frac{1}{N_n}\big\|M_n\varphi^{(n)}\big\|^2
        \leq
        \frac{1}{N_n}\big\|M_n\varphi^{(n)}\big\|_F^2
        \leq
        \frac{1}{N_n}\sum_{i_n}\sum_{j_n}
          {W^{(n)}_{i_n,j_n}}
         \left\|\varphi_{i_n}^{(n)} - {\varphi}_{j_n}^{(n)}\right\|^2.
    \end{align*}

Use the triangle inequality 
and $(a + b + c)^2 \leq 3(a^2 + b^2 + c^2)$ to bound this by 
\begin{align*}
    \frac{1}{N_n}\sum_{i_n}\sum_{j_n} 
          {W^{(n)}_{i_n,j_n}} 
          O\left(
          \left\| \varphi_{i_n}^{(n)} - \widehat{\varphi}_{i_n}^{(n)} \right\|^2
    + 
    \left\| \varphi_{j_n}^{(n)} - \widehat{\varphi}_{j_n}^{(n)} \right\|^2
    + 
    \left\| \widehat{\varphi}_{i_n}^{(n)} - \widehat{\varphi}_{j_n}^{(n)} \right\|^2
         \right).
\end{align*}
Take the first of these terms, 
\begin{align*}
          \frac{1}{N_n}\sum_{i_n}\sum_{j_n} 
          {W^{(n)}_{i_n,j_n}} 
          \left\| \varphi_{i_n}^{(n)} - \widehat{\varphi}_{i_n}^{(n)} \right\|^2
          &=
          \frac{1}{N_n}\sum_{i_n}
          \left\| \varphi_{i_n}^{(n)} - \widehat{\varphi}_{i_n}^{(n)} \right\|^2
          \sum_{j_n} 
          {W^{(n)}_{i_n,j_n}} 
          = 
          O_p(C^{-2}_n)
\end{align*}
Likewise, the second term is $O_p(C^{-2}_n)$, since $\max_{j_n}\sum_{i_n}{W^{(n)}_{i_n,j_n}} = O_p(1)$ by Assumption~\ref{ass:regCondKer}.(iv).

From Assumption~\ref{ass:Kernels}.(v) $k\left(\big\| \widehat{\varphi}_{i_n}^{(n)} - \widehat{\varphi}_{j_n}^{(n)} \big\|/h_n\right) = 0$
if $\big\| \widehat{\varphi}_{i_n}^{(n)} - \widehat{\varphi}_{j_n}^{(n)} \big\| > Uh_n$.
Hence,
\begin{align}
    \label{proof:weightedNorms}
    \frac{1}{N_n}\sum_{i_n}\sum_{j_n} 
          {W^{(n)}_{i_n,j_n}} 
    \left\| \widehat{\varphi}_{i_n}^{(n)} - \widehat{\varphi}_{j_n}^{(n)} \right\|^2
    &\leq 
    U^2h_n^2 
    \frac{1}{N_n}\sum_{i_n}\sum_{j_n} 
          {W^{(n)}_{i_n,j_n}} 
    =
    O(h_n^2) 
\end{align}

Hence, for each $n\in\mathcal{M}^\prime$,
\begin{align*}
    \frac{1}{N_n}\sum_{i_n}\sum_{j_n} 
          {W^{(n)}_{i_n,j_n}} 
         \left\|\varphi_{i_n}^{(n)} - {\varphi}_{j_n}^{(n)}\right\|^2
    \leq
    \left(O_p\left(C_n^{-2}\right) + O_p\left(h_n^{2}\right) \right).
\end{align*}
Substituting these bounds into \eqref{proof:tuckerBound} gives $N^{-1/2}\|\widetilde{\boldsymbol{\mathcal{A}}}\|_F = O_p\big(\prod_{n\in\mathcal{M}^\prime}(C_n^{-1} + h_n)\big)$, which with the bound on $\norm{\omega_N}$ forms the statement of the result.
Any basis of the column space of $\mathcal{A}_{(n)}$ that satisfies the eigenvalue condition can serve as $\varphi^{(n)}$, so the rotation $H^k_n$ in Lemma~\ref{lemma:cons_proxy} is absorbed into $\varphi^{(n)}$.

\end{proof}

Theorems~\ref{thm:AsyNormHom}-\ref{thm:AsyNormHet} are not proven since Theorem~\ref{thm:AsyNormHetACCC} is more general.

\begin{proof}[\bf Proof of Theorem~\ref{thm:AsyNormHetACCC}]
    Let $N:= \prod_{n = 1}{N_n}$. Begin with \eqref{eqn:InfError}, 
    \begin{align}\label{eqn:proofIC}
        \widehat{\beta}_{IC} - \beta^0 = \widehat{\Omega}_X^{-1}N^{-1}B
    + \widehat{\Omega}_X^{-1}{\rm vec}_K(\widehat{\boldsymbol{\eta}})^\prime{\rm vec}(\boldsymbol{\varepsilon})
    \end{align}
    where, the bias term, $B$ with,  
    \begin{align*}
        B :=
        {\rm vec}_K(\widehat{\boldsymbol{\eta}})^\prime
        {\rm vec}
        (\boldsymbol{\mathcal{A}}  - \widehat{\boldsymbol{\mathcal{A}}}).
    \end{align*}
    To proceed, need to show $N^{-1}B = o_p(N^{-1/2})$. 
\begin{align*}
    \frac{1}{N}B &= 
    \frac{1}{N}{\rm vec}_K(\boldsymbol{X} - \widehat{\boldsymbol{\Gamma}}_X)^\prime
        {\rm vec}
        (\boldsymbol{\mathcal{A}}  - \widehat{\boldsymbol{\mathcal{A}}})
        \\
        &= 
        \frac{1}{N}{\rm vec}_K({\boldsymbol{\Gamma}}_X - \widehat{\boldsymbol{\Gamma}}_X)^\prime
        {\rm vec}
        (\boldsymbol{\mathcal{A}}  - \widehat{\boldsymbol{\mathcal{A}}})
        +
        \frac{1}{N}{\rm vec}_K(\boldsymbol{\eta})^\prime
        {\rm vec}
        (\boldsymbol{\mathcal{A}}  - \widehat{\boldsymbol{\mathcal{A}}})
\end{align*}
Hence, 
\begin{align*}
     \frac{1}{\sqrt{N}}B &= 
     O_p(\sqrt{N}\xi_X\xi_{\mathcal{A}})\iota_K
    + \frac{1}{\sqrt{N}}{\rm vec}_K(\boldsymbol\eta)^\prime {\rm vec}
    (\boldsymbol{\mathcal{A}}  - \widehat{\boldsymbol{\mathcal{A}}}).
\end{align*}
The first term is $o_p(1)$ since $\xi_{\mathcal{A}} = O_p(N^{-1/2})$, and $\xi_X = o_p(1)$. 
For the second term,
\begin{align}\label{eqn:proofVarIC}
    \textrm{var}\left(\frac{1}{\sqrt{N}}{\rm vec}_K(\boldsymbol\eta)^\prime {\rm vec}
    (\boldsymbol{\mathcal{A}}  - \widehat{\boldsymbol{\mathcal{A}}})
    \right)
    =
    o(1)
\end{align}
and $\mathbb{E}\left({N}^{-1/2}{\rm vec}_K(\boldsymbol\eta)^\prime {\rm vec}
(\boldsymbol{\mathcal{A}}  - \widehat{\boldsymbol{\mathcal{A}}})\right) = 0$. 
The zero variance derivation is similar to the derivation below for ${\rm vec}_K({\boldsymbol{\Gamma}}_{X}
    - \widehat{\boldsymbol{\Gamma}}_{X})^\prime{\rm vec}(\boldsymbol{\varepsilon})$ term from Assumption~\ref{ass:infHighLevel}.
The last term in \eqref{eqn:proofIC} is, 
\begin{align*}
    \frac{1}{\sqrt{N}}{\rm vec}_K(\widehat{\boldsymbol{\eta}})^\prime{\rm vec}(\boldsymbol{\varepsilon})
    =
    \frac{1}{\sqrt{N}}{\rm vec}_K
    ({\boldsymbol{\Gamma}}_{X}
    - \widehat{\boldsymbol{\Gamma}}_{X})^\prime{\rm vec}(\boldsymbol{\varepsilon})
    +
    \frac{1}{\sqrt{N}}{\rm vec}_K({\boldsymbol{\eta}})^\prime{\rm vec}(\boldsymbol{\varepsilon}).
\end{align*}
The second term is asymptotically normally distributed by Assumption~\ref{ass:HeteroCorrCLT}. 
The first term is mean zero and has variance, 
\begin{align*}
    &\mathbb{E}\left[ 
    \frac{1}{N}\sum_{i_1,\dots,i_d}
    \left({\boldsymbol{\Gamma}}_{X, i_1,\dots,i_d} - \widehat{\boldsymbol{\Gamma}}_{X,i_1,\dots,i_d}\right)
    \left({\boldsymbol{\Gamma}}_{X, i_1,\dots,i_d} - \widehat{\boldsymbol{\Gamma}}_{X,i_1,\dots,i_d}\right)^\prime
    \varepsilon_{i_1,\dots,i_d}^2
    + \right. \dots \\
    &\left.
    \dots + 
    \frac{1}{N}
    \sum_{i_1^\prime\dots i_d^\prime\neq i_1\dots i_d}
    \sum_{i_1,\dots,i_d}
    \left({\boldsymbol{\Gamma}}_{X, i_1,\dots,i_d} - \widehat{\boldsymbol{\Gamma}}_{X,i_1,\dots,i_d}\right)
    \left({\boldsymbol{\Gamma}}_{X, i_1^\prime,\dots,i_d^\prime} - \widehat{\boldsymbol{\Gamma}}_{X,i_1^\prime,\dots,i_d^\prime}\right)^\prime
    \varepsilon_{i_1,\dots,i_d} 
    \varepsilon_{i_1^\prime,\dots,i_d^\prime}
    \right]
\end{align*}
The first term is $o(1)$ from Assumption~\ref{ass:inf3}-~\ref{ass:inferenceInd} and bounded $\mathbb{E}\varepsilon_{i_1,\dots,i_d}^2$. 
Use shorthand index notation $i$ in place of $i_1,\dots,i_d$.
From Assumption~\ref{ass:inferenceInd}, for each $k, k^\prime$, 
\begin{align*}
    \frac{1}{N}
    &\sum_{j\neq i}
    \sum_{i}
    \mathbb{E}\left[ 
    \left({\boldsymbol{\Gamma}}_{X_{k}, i} - \widehat{\boldsymbol{\Gamma}}_{X_{k},i}\right)
    \left({\boldsymbol{\Gamma}}_{X_{k^\prime}, j} - \widehat{\boldsymbol{\Gamma}}_{X_{k^\prime},j}\right)
    \right]
    \mathbb{E}\left[ 
    \varepsilon_{i} 
    \varepsilon_{j}
    \right]
    \\
    &\leq 
    \left(
    \frac{1}{N}
    \sum_{j\neq i}
    \sum_{i}
    \left(\mathbb{E}\left[ 
    \left({\boldsymbol{\Gamma}}_{X_{k}, i} - \widehat{\boldsymbol{\Gamma}}_{X_{k},i}\right)
    \left({\boldsymbol{\Gamma}}_{X_{k^\prime}, j} - \widehat{\boldsymbol{\Gamma}}_{X_{k^\prime},j}\right)
    \right]\right)^2
    \right)^{1/2}
    \left(
    \frac{1}{N}
    \sum_{j\neq i}
    \sum_{i}
    \left(\mathbb{E}\left[ 
    \varepsilon_{i} 
    \varepsilon_{j}
    \right]\right)^2
    \right)^{1/2}
    \\
    &\leq
    \left(
    \frac{1}{N}
    \sum_{j\neq i}
    \sum_{i}
    \mathbb{E}\left[ 
    \left({\boldsymbol{\Gamma}}_{X_{k}, i} - \widehat{\boldsymbol{\Gamma}}_{X_{k},i}\right)^2
    \right]
    \mathbb{E}\left[ 
    \left({\boldsymbol{\Gamma}}_{X_{k^\prime}, j} - \widehat{\boldsymbol{\Gamma}}_{X_{k^\prime},j}\right)^2
    \right]
    \right)^{1/2}
    \left(
    \frac{1}{N}
    \sum_{j\neq i}
    \sum_{i}
    \sigma_{ij}^2
    \right)^{1/2}
    \\
    &=
    O_p\left(N^{1/2} \xi_X^2 \xi_{\sigma,N}\right).
\end{align*}
Notation comes from Assumption~\ref{ass:infHighLevel}, which also implies $O_p\left(N^{1/2} \xi_X^2 \xi_{\sigma,N}\right) = o_p(1)$. 
Similar derivations get to \eqref{eqn:proofVarIC}. 

Finally, $\widehat{\Omega}_X^{-1} = (\mathbb{E}(\eta_{i_1,\dots,i_d}\eta_{i_1,\dots,i_d}^\prime) + o_p(1))^{-1}$. Assumption~\ref{ass:HeteroCorrCLT} gives the asymptotic distribution for $N^{-1/2}{\rm vec}_K({\boldsymbol{\eta}})^\prime{\rm vec}(\boldsymbol{\varepsilon})$, which completes the proof. 

\end{proof}

\subsection{Supplementary results}\label{sect:AppendixSupp}

To verify Assumption~\ref{ass:regCondKer}.(ii), a sample splitting device is constructed to remove dependence between fixed-effect estimates and idiosyncratic terms. The approach is similar to the sample splitting device used in \cite{freeman2023linear}.
Two dimensions are used for simplicity. 
Assume weak dependence in the noise terms, $\varepsilon,\eta$, where $\mathcal{P}$ is the probability law. 
Two mixing conditions are considered here for comparison: the strong, or  $\alpha$-mixing, condition and the uniform $\phi$-mixing condition. 
There are potentially other conditions between these two in strictness which admit different levels of boundedness in the distribution of noise terms. 

Let $\mathcal{A}$ and $\mathcal{B}$ denote $\sigma$-fields on the probability space that $\varepsilon$ lies in. 
Let $\varepsilon_{i_1i_2} \in \mathcal{A}$ and $\varepsilon_{i_1'i_2'} \in \mathcal{B}$. 
Strong mixing for $\varepsilon$ with $\delta_1,\delta_2>0$ is restricted to,
\begin{align}\label{eqn:StrongMixing}
    \sup_{A\in \mathcal{A}, B \in \mathcal{B}}
    \left\{|\mathcal{P}[{A} \cap {B}]
    -
    \mathcal{P}[{A} ] \mathcal{P}[{B}]|: |i_1 - i_1'|\geq a_1, |i_2 - i_2'|\geq a_2\right\}
    = o(a_1^{-\delta_1}\cdot a_2^{-\delta_2}).
\end{align}
Strong mixing can be used for covariance bounds on $\varepsilon$ with sufficient decay as $a_1,a_2\to \infty$ when $\varepsilon$ is distributed in a bounded set. 
The uniform $\phi$-mixing condition that admits $\varepsilon$ distributed on an unbounded set, with regularity, is, 
\begin{align}\label{eqn:UniformMixing}
    \sup_{\substack{A \in \mathcal{A},
    B \in \mathcal{B}\\\mathcal{P}[B]\neq 0}}
    \left\{|\mathcal{P}[A|B]
    -
    \mathcal{P}[A] :
    |i_1 - i_1'|\geq a_1, |i_2 - i_2'|\geq a_2\right\}
    = o(a_1^{-\delta_1}\cdot a_2^{-\delta_2})
\end{align}

Mixing conditions can be motivated in, for example, time or geo-spatially distanced indices.\footnote{For geo-spatial indices, these could always be replaced with a distance measure to be used in constructing the proposed sample split. }
Conditions \eqref{eqn:StrongMixing}-\eqref{eqn:UniformMixing} are trivially satisfied with, e.g., full independence, but also admit more general classes of weak dependence. 

Weights are estimated using different partitions to ensure independence between $\varepsilon$ and weights. 
Denote $\mathcal{S} = \{1,\dots, N_1\}\times\{1,\dots,N_2\}$. 
Take some $b>0$, such that $b = O(\min\{N_1,N_2\})$,
\begin{align*}
    \begin{split}
        \mathcal{S}_1 &= \{1,\dots, N_1/2-b\}\times\{1,\dots,N_2/2-b\}\\
        \mathcal{S}_2 &= \{1,\dots, N_1/2-b\}\times\{N_2/2+b,\dots,N_2\}\\
        \mathcal{S}_3 &= \{N_1/2 +b,\dots, N_1\}\times\{1,\dots,N_2/2-b\}\\
        \mathcal{S}_4 &= \{N_1/2 +b,\dots, N_1\}\times\{N_2/2+b,\dots,N_2\}
    \end{split}
\end{align*}
Then, estimate proxies for weights $(i_1,i_2)\in\mathcal{S}_k$ from partition $(i_1, i_2)\in \mathcal{S}\backslash \mathcal{S}_k$ for all $k$. 
For example, to estimate $\lambda_{i_1}, f_{i_2}$ for $(i_1,i_2)\in\mathcal{S}_1$, take 
$\hat \lambda_{i_1}$ from $\mathcal{S}_2$ and $\hat f_{i_2}$ from $\mathcal{S}_3$. 
Then, for $b\rightarrow\infty$, \eqref{eqn:StrongMixing}-\eqref{eqn:UniformMixing} implies $\varepsilon_{i_1 i_2 }$ are asymptotically independent of weights. 
For any dimension that admits full independence in idiosyncratic terms there is no need to create this separation. For example, if $i_1$ is an independently sampled cross section, the first partition would become, 
\begin{align*}
        \mathcal{S}_1 &= \{1,\dots, N_1/2\}\times\{1,\dots,N_2/2-b\}, && \textrm{etc.}
\end{align*}

Take $(i_1,i_2)\in \mathcal{S}_k$ and $(i_1^\prime, i_2^\prime)\in \mathcal{S}_k$, i.e. from the same partition. Denote $\mathcal{F}_{-k}$ as the filtration of data $Y,X$, for $(i_1,i_2)\in \mathcal{S}\backslash \mathcal{S}_k$.
This sample split then admits negligible differences between considering $\mathbb{E}[\varepsilon_{i_1i_2} \varepsilon_{i_1^\prime i_2^\prime}|\mathcal{F}_{-k}]$ versus $\mathbb{E}[\varepsilon_{i_1i_2} \varepsilon_{i_1^\prime i_2^\prime}]$. 
To show, notice in the case of mixing condition \eqref{eqn:StrongMixing} with $\varepsilon$ distributed on a bounded set,  
\begin{align*}
    \Big|\mathbb{E}&[\varepsilon_{i_1i_2} \varepsilon_{i_1^\prime i_2^\prime}|\mathcal{F}_{-k}]
    - 
    \mathbb{E}[\varepsilon_{i_1i_2} \varepsilon_{i_1^\prime i_2^\prime}]\Big|
    = 
    \Big|\int_{\Omega_\varepsilon\times \Omega_\varepsilon} uv\,\, d\left[\mathcal{P}_{EE'|F_{-k}} - \mathcal{P}_{EE'} \right]  \Big|
    \\
    &\leq 
    \int_{\Omega_\varepsilon\times \Omega_\varepsilon}  |uv| d\Big|\mathcal{P}_{EE'|F_{-k}} - \mathcal{P}_{EE'}  \Big| 
    \leq 
    \int_{\Omega_\varepsilon\times \Omega_\varepsilon}  |uv| d \sup \Big|\mathcal{P}_{EE'|F_{-k}} - \mathcal{P}_{EE'}  \Big|  = o_p(b^{-\delta_1 - \delta_2})
\end{align*}
where this uses that any $(i_1,i_2)\in \mathcal{S}_k$ is separated from data in the filtration $\mathcal{F}_{-k}$ by at least $b$ observations in both indices, then applies the condition in \eqref{eqn:StrongMixing}. 

For the uniform mixing case use $\mathcal{A}$ to denote the related $\sigma$-field of the event $\varepsilon_{i_1i_2} \varepsilon_{i_1^\prime i_2^\prime}$ where $(i_1,i_2)\in \mathcal{S}_k$, $(i_1^\prime, i_2^\prime)\in \mathcal{S}_k$. Denote $\mathcal{F}_{-k}$ as the filtration of data $Y,X$, for $(i_1,i_2)\in \mathcal{S}\backslash \mathcal{S}_k$, as above. 
Then, from \eqref{eqn:UniformMixing}, which considers $\mathcal{P}_{\mathcal{F}_{-k}}\neq 0$, arguments from \cite{rio2013inequalities} show, 
\begin{align*}
    \Big|\mathbb{E}&[\varepsilon_{i_1i_2} \varepsilon_{i_1^\prime i_2^\prime}|\mathcal{F}_{-k}]
    - 
    \mathbb{E}[\varepsilon_{i_1i_2} \varepsilon_{i_1^\prime i_2^\prime}]\Big|
    \leq  
    \int |A| \,\,d |\mathcal{P}_{\mathcal{A}|\mathcal{F}_{-k}} - \mathcal{P}_{\mathcal{A}}|
    =  
    \int |A| \,\,d \mathcal{P}_{\mathcal{A}}/ \mathcal{P}_{\mathcal{F}_{-k}}|\mathcal{P}_{\mathcal{F}_{-k}|\mathcal{A}} - \mathcal{P}_{\mathcal{F}_{-k}}|
    \\
    &\leq 
    \left(\int |A|^2 \,\,d \mathcal{P}_{\mathcal{A}} \right)^{1/2}
    \left(\int d 
    |\mathcal{P}_{\mathcal{F}_{-k}|\mathcal{A}} - \mathcal{P}_{\mathcal{F}_{-k}}|^2
    \mathcal{P}_{\mathcal{A}}/ \mathcal{P}_{\mathcal{F}_{-k}}^2
    \right)^{1/2}
    \leq 
    \mathbb{E}[(\varepsilon_{i_1i_2} \varepsilon_{i_1^\prime i_2^\prime})^2]^{1/2} o_p(b^{-\delta_1 - \delta_2})
\end{align*}
Hence, for bounded $\mathbb{E}[\varepsilon_{i_1i_2}^4]$ the term is bounded $o_p(b^{-\delta_1 - \delta_2})$.

Hence, we can confirm a conditional covariance bound with an unconditional one:
\begin{align*}
    \frac{1}{N_1N_2} \sum_{(i_1, i_2)\in \mathcal{S}_k}\sum_{i_1', i_2'\in \mathcal{S}_k}\mathbb{E}[\varepsilon_{i_1 i_2}\varepsilon_{i_1' i_2'}|\mathcal{F}_{-k}]
    &= 
    \frac{1}{N_1N_2} \sum_{(i_1, i_2)\in \mathcal{S}_k}\sum_{i_1', i_2'\in \mathcal{S}_k}
    (\mathbb{E}[\varepsilon_{i_1 i_2}\varepsilon_{i_1' i_2'}]
    + 
    o_p(b^{-\delta_1-\delta_2}))
    \\
    &=
    o_p(1)
    +
    \frac{1}{N_1N_2} \sum_{(i_1, i_2)\in \mathcal{S}_k}\sum_{i_1', i_2'\in \mathcal{S}_k}
    \mathbb{E}[\varepsilon_{i_1 i_2}\varepsilon_{i_1' i_2'}], 
\end{align*}
whenever $\delta_1,\delta_2 >1$ and $o_p(N_1N_2\cdot b^{-\delta_1 - \delta_2})) = o_p(1)$. 
Say $N_1\propto N_2$ and $\delta := \delta_1 = \delta_2$. 
Then the condition is $b \gtrsim (N_1N_2)^{1/2\delta} \approx (N_1)^{1/\delta}$. 
Stronger mixing processes, from higher $\delta$, imply the separation between partitions defined by $b$ can be smaller.
This fits intuition: the faster the random variable mixes, the less these partitions need to be separated. 

For the below Lemma, consider one dimension for simplicity. 
Also abstract away from the mixing conditions, and refer to the arguments made already that it is without loss of generality to consider expectations of noise terms without conditioning on the filtrations used to generate weights.
That is, that it can be assumed weights are independent of noise terms. 

\begin{lemma}[Verification of Assumption~\ref{ass:regCondKer}.(ii)]\label{lemma:regCondKer}
    Maintain Assumptions~\ref{ass:norms}-\ref{ass:Exog}, and 
    $\varepsilon_i$ independent of fixed-effects and weights,
    with bandwidth $h\gtrsim N^{-1/2}$. Impose in $\plim_{N\rightarrow\infty}$,\footnote{Assume without loss that $X_i$ is mean zero. }
    \begin{align*}
        &\Big(\frac{1}{N}\sum_{i}\sum_{i^\prime}\big( \mathbb{E}\big[\varepsilon_i \varepsilon_{i^\prime }|X\big]\big)^2\Big)^{1/2};\,\,
        \frac{1}{N}\sum_{i}\sum_{i^\prime} 
        \mathbb{E}\big[\varepsilon_i \varepsilon_{i^\prime }|X\big];\,\,
        \frac{1}{N}\sum_{i}\sum_{i^\prime} \mathbb{E}\big[{X}_i{X}_{i^\prime}\big]
        \leq M <\infty,
    \end{align*}
    If weights for each $i$ such that $\sum_j W_{i,j}^2 = O_p({N}h)^{-1}$;
    \footnote{Assumption~\ref{ass:regCondKer}.(iii) and standard nonparametric arguments attain this pointwise over $i$ for $h > cN^{-1/2}$ when weights are functions of scalars -- see Section~\ref{sect:AppendixIterative} for the case using iterative scalar weights. }
    then, 
    $N^{-1}\sum_i \check{X}_i \check{\varepsilon}_i =O_p(N^{-1/2})$.
\end{lemma}
\begin{proof}[\bf  Proof of Lemma~\ref{lemma:regCondKer}]
    The term $N^{-1}\sum_i \check{X}_i \check{\varepsilon}_i$ can be written, 
    \begin{align*}
        \frac{1}{N}\sum_i \check{X}_i \check{\varepsilon}_i
        =
        \frac{1}{N}\sum_i \left[ 
        {X}_i{\varepsilon}_i
        - \sum_jW_{i,j}{X}_j{\varepsilon}_i
        - \sum_\ell W_{i,\ell}{X}_i{\varepsilon}_\ell
        + \sum_j\sum_\ell W_{i,j}W_{i,\ell}{X}_j{\varepsilon}_\ell
        \right]
    \end{align*}
    The term ${N^{-1}}\sum_i{X}_i{\varepsilon}_i = O_p(N^{-1/2})$ from Assumption~\ref{ass:Exog}. 
    The second term, 
    \begin{align*}
        Var&\Big(\frac{1}{N}\sum_i \sum_jW_{i,j}{X}_j{\varepsilon}_i\Big) 
        =
        \frac{1}{N^2}\mathbb{E}\left[\sum_i \sum_{i^\prime} 
        \mathbb{E}\big[\varepsilon_i \varepsilon_{i^\prime}|X] 
        \mathbb{E}\Big[\sum_j  W_{ij}  {X}_j \sum_{j^\prime} W_{i^\prime j^\prime}{X}_{j^\prime}\big|X\Big] \right]
        \\
        &\leq 
        \frac{M}{N^{3/2}}\mathbb{E}\left[\Big(
        \mathbb{E}\Big[
        \sum_i \Big(\sum_j  W_{ij}  {X}_j \Big)^2
        \sum_{i^\prime} \Big(\sum_{j^\prime} W_{i^\prime j^\prime}{X}_{j^\prime}\Big)^2  \Big|X\Big]\Big)^{1/2}\right]
        \\
        &\leq
        \frac{M}{N^{3/2}}\mathbb{E}\left[\Big(
        \mathbb{E}\Big[\Big(
        \sum_j \sum_{j'}{X}_j {X}_{j'} \sum_i W_{ij} W_{ij'}\Big)^2
        \Big|X\Big]
        \Big)^{1/2}\right]
        \lesssim 
        \frac{M^2}{N^{3/2}h} = O(N^{-1}),
    \end{align*}
    by Jensen's and Cauchy-Schwarz inequality. 
    The third term follows similarly. 
    Last term, 
    \begin{align*}
        Var\Big(\frac{1}{N}&\sum_{i j\ell} W_{i,j}W_{i,\ell}{X}_j{\varepsilon}_\ell\Big)
        = 
        \frac{1}{N^2}\sum_{ j\ell} \sum_{ j'\ell'} 
        \mathbb{E}\left[
        {X}_j{\varepsilon}_\ell
        {X}_{j'}{\varepsilon}_{\ell'}
         \sum_i W_{i,j}W_{i,\ell}
        \sum_{i'}  W_{i',j'}W_{i',\ell'}
        \right]
        \\
        &\lesssim 
        \frac{1}{N^2}\sum_{ j\ell} \sum_{ j'\ell'} 
        \mathbb{E}\left[
        {X}_j{\varepsilon}_\ell
        {X}_{j'}{\varepsilon}_{\ell'}
         (Nh)^{-2}
        \right]
        = 
        \frac{1}{N^2}\sum_{j j'} 
        \mathbb{E}\left[
        {X}_j
        {X}_{j'}
        \sum_{\ell\ell'}\mathbb{E}[{\varepsilon}_\ell{\varepsilon}_{\ell'}|X]
         O(Nh)^{-2}
        \right]
        \\
        &\leq
        (Nh)^{-2}O_p(M^2) = O(N^{-1}).
    \end{align*}
\end{proof}

\section{Iterative Estimator}\label{sect:AppendixIterative}

\renewcommand{\thetheorem}{B.\arabic{theorem}} 
\setcounter{theorem}{0}
\renewcommand{\thelemma}{B.\arabic{lemma}} 
\setcounter{lemma}{0}
\renewcommand{\theproposition}{B.\arabic{proposition}} 
\setcounter{proposition}{0}
\renewcommand{\theequation}{B.\arabic{equation}} 
\setcounter{equation}{0}
\renewcommand{\theassumption}{B.\arabic{assumption}} 
\setcounter{assumption}{0}

The following allows for high-dimensional set of proxy measures used to form projection weights.  
Consider the higher-order singular value decomposition (HOSVD), 
\begin{align*}
    \mathcal{A} = \mathcal{Q}\times (U^{(1)}, \dots, U^{(d)})
\end{align*}
where $\mathcal{Q}\in\mathbb{R}^{\bigtimes_{n=1}^d r_n}$ and unitary $N_n^{-1/2}U^{(n)} \in \mathbb{R}^{N_n\bigtimes r_n}$, where $r_n$ is the $n$-th compononent of the multilinear rank.\footnote{Unitary $N_n^{-1/2}U^{(n)}$ instead of unitary $U^{(n)}$ simply rescales singular values $\mathcal{Q}$. }
$U^{(n)}$ are considered columnwise mean zero without loss. 
The iterative estimator is a backfitting algorithm that works as follows. 

\begin{enumerate}
    \item Obtain proxies for $U^{(n)}$, called $\widehat{U}^{(n)}$ for each $n = 1,\dots, d$.
    Normalise these to be columnwise mean zero. 
    Use these to form weights, for $m_n = 1,\dots,r_n$;
    \begin{align}\label{eqn:AppKernelWeights}
	S_{nm_n,i_nj_n} 
    := 
    \frac{k\left(\frac{1}{h_n} \left|{\widehat{U}^{(n)}_{i_n,m_n} - \widehat{U}^{(n)}_{j_n,m_n}}\right| \right)}
    {\sum_{i^\prime_n = 1}^{N_n} k\left(\frac{1}{h_n} \left|{\widehat{U}^{(n)}_{i_n,m_n} - \widehat{U}^{(n)}_{i^\prime_n,m_n}}\right|\right)}
    \quad\quad
    \textrm{ for each }
    n = 1,\dots, d. 
\end{align}
\begin{enumerate}[(i).]
    \item For each $m_n = 1,\dots,r_n$ difference out weighted means from $\boldsymbol{Y}$ and $\boldsymbol{X}_k$ for each $k = 1,\dots, K$ along each dimension, $n = 1,\dots, d$ as follows, 
    \begin{align}\label{eqn:AppKernelDiff}
    \widetilde{\mathcal{A}}_{(\ell)} 
    &= \widetilde{\mathcal{A}}_{(\ell-1)} \times \Bigg(\prod_{m_1 = 1}^{r_1}\big(\mathbb{I} - S_{1m_1}\big), \dots, \prod_{m_d = 1}^{r_d}\big(\mathbb{I} - S_{dm_d}\big)\Bigg)
    \end{align}
    \item Iterate over $\ell$ until convergence. Call the convergent value $\widetilde{\mathcal{A}}$. 
\end{enumerate}
    \item Perform pooled OLS of $\widetilde{\boldsymbol{Y}}$ on $\widetilde{\boldsymbol{X}}$. Call the estimator $\widehat\beta_{IW}$, for iterative weighted. 
\end{enumerate}

Step 1(ii) is the backfitting component, see e.g. \cite{opsomer1999root,opsomer2000asymptotic}. 
In the simulations and application the proxies are re-estimated at each iteration from the flattened residual $\boldsymbol{Y} - \boldsymbol{X}\cdot\widehat{\beta}$, with $\widehat{\beta}$ updated by least squares, and the iterations stop once $\widehat{\beta}$ converges. 
Label $W$ as the weights implied by the backfitting algorithm. From \cite{opsomer1999root} $W$ admit the additive form $W_n = \sum_{m_n=1}^{r_n} W_{nm_n}$ such that, 
\begin{align*}
    \widetilde{\mathcal{A}} = 
    \mathcal{Q}\times ((\mathbb{I} - W_1)U^{(1)}, \dots, (\mathbb{I} - W_d)U^{(d)}).
\end{align*}
\cite{opsomer2000asymptotic} show these are asymptotically equivalent to the smoother weights $S$: $W_{nm_n} = S_{nm_n} + O_p(\iota_{N_n}\iota_{N_n}'/N_n)$, which significantly simplifies the theory. 

Proof techniques from previous results follow from the representation in \eqref{eqn:AppKernelDiff}. 
As proposed in the main text, proxies $\widehat{U}^{(n)}$ can be estimated using the matrix method with each dimension used as rows. 
Proposition A.1 in \cite{Bai2009} can be used to bound the estimation error of these proxies with respect to the true $U^{(n)}$. 
The fact Proposition A.1 in \cite{Bai2009} is an up-to-rotations result does not affect convergence rates, and is discussed below the proof of Proposition~\ref{lemma:consIterKernel}. 
Hence, for the statement of the result, these rotations are ignored.

\begin{theorem}[Upper bound on iterative estimator]\label{lemma:consIterKernel}
    Make Assumption~\ref{ass:Kernels}, and Assumption~\ref{ass:regCondKer} specifically for the iterative transformation.
    For the HOSVD of fixed-effects, 
    impose $\frac{1}{N}\|\mathcal{A}\|_F^2 = O_p(1)$ with $N = \prod_n N_n$, and let
    $\frac{1}{N_{n}}\sum_{i_{n^*}}\left\| {U}^{(n)}_{i_{n}} - \widehat{U}^{(n)}_{i_{*}}\right\|^2 = O_p(C_{n}^{-2})$.
    Additionally, assume $\widehat{U}^{(n)}_{m_n}$ are mutually independent over $m_n$.\footnote{Independence significantly simplifies the proof without the need for iterative arguments, but can be relaxed, see e.g. \cite{opsomer1999root,opsomer2000asymptotic}. } 
    Then, 
    \begin{align*}
        \norm{\widehat{\beta}_{IW} - \beta^0} 
        = 
        \prod_{n=1}^d O_p\Big(C_n^{-1} + \sqrt{r_n}\,h_n + \sqrt{r_n}(r_n-1) \big(N_n \sqrt{h_n}\big)^{-1} \Big)
        +
        O_p\left(\prod_{n = 1}^d\frac{1}{\sqrt{N_{n}}} \right)
        .
    \end{align*}
    For $N_n^{-1}\lesssim h_n \lesssim C_n^{-1}$ and fixed $r_n$ this reduces to
    \begin{align*}
        \norm{\widehat{\beta}_{IW} - \beta^0}
        =
        O_p\left(\prod_{n}
        C_{n}^{-1}
        \right)
        + 
        O_p\left(\prod_{n = 1}^d\frac{1}{\sqrt{N_{n}}} \right)
        .
    \end{align*}
\end{theorem}

\begin{proof}[\bf Proof of Proposition~\ref{lemma:consIterKernel}]

Weights $W_{nm_n} = S_{nm_n} + O_p(\iota_{N_n}\iota_{N_n}'/N_n)$. 
For each column, $\ell_n$ of $(\mathbb{I} - \sum_{m_n=1}^{r_n}W_{n,m_n}) U^{(n)}$ notice, 
\begin{align*}
    (\mathbb{I} - \sum_{m_n=1}^{r_n}W_{n,m_n}) U^{(n)}_{\ell_n}
    &= 
    (\mathbb{I} - W_{n,\ell_n}) U^{(n)}_{\ell_n}
    - \sum_{m_n\neq \ell_n}W_{n,m_n} U^{(n)}_{\ell_n}
    \\
    &= 
    (\mathbb{I} - S_{n,\ell_n}) U^{(n)}_{\ell_n}
    - \Big(r_nO_p(\mathbbm{1}/N_n) + \sum_{m_n\neq \ell_n}S_{n,m_n}\Big) U^{(n)}_{\ell_n}.
\end{align*}
It is without loss to consider mean zero $U^{(n)}_{\ell_n}$. 
Hence, $r_nO_p(\mathbbm{1}/N_n) U^{(n)}_{\ell_n} = \boldsymbol{0}_{N_n}$. 
Call $\widetilde{A}_{1,\ell_n}^{(n)} := (\mathbb{I} - S_{n,\ell_n}) U^{(n)}_{\ell_n}$, and $\widetilde{A}_{2,\ell_n}^{(n)} := \sum_{m_n\neq \ell_n}S_{n,m_n} U^{(n)}_{\ell_n}$.
Then, 
\begin{align*}
    \widetilde{\boldsymbol{A}} = 
    {\mathcal{Q}}\times 
    (\widetilde{A}_{1}^{(1)} + \widetilde{A}_{2}^{(1)} , \dots, \widetilde{A}_{1}^{(d)} + \widetilde{A}_{2}^{(d)} )
\end{align*}
Consider first, the component, 
\begin{align*}
    \frac{1}{N} \left|\left\langle \widetilde{\boldsymbol{X}}, \widetilde{\boldsymbol{A}}_1 \right\rangle_F\right|
    \leq 
    \frac{1}{N} \left\langle \widetilde{\boldsymbol{X}}, \widetilde{\boldsymbol{X}} \right\rangle_F^{1/2}
    \left\langle \widetilde{\boldsymbol{A}}_1, \widetilde{\boldsymbol{A}}_1 \right\rangle_F^{1/2}
    = 
    O_p(1) 
    \frac{1}{\sqrt{N}}
    \left\| \widetilde{\boldsymbol{A}}_1\right\|_F,
\end{align*}
where $\widetilde{\boldsymbol{A}}_1 := \widetilde{\boldsymbol{S}}\times (\widetilde{A}_{1}^{(1)}, \dots, \widetilde{A}_{1}^{(d)} )$.
Use the HOSVD 
\begin{align*}
    \frac{1}{\sqrt{N}}
    \left\| \widetilde{\boldsymbol{A}}_1\right\|_F 
    &= 
    \frac{1}{\sqrt{N}}
    \left\| \mathcal{Q}\times (\widetilde{A}_{1}^{(1)}, \dots, \widetilde{A}_{1}^{(d)}\right\|_F =
    \frac{1}{\sqrt{N}}
    \left\| \widetilde{A}_{1}^{(1)} \mathcal{Q}_{[1]} \left[\bigotimes_{n \neq 1} \widetilde{A}_{1}^{(d)}\right]'\right\|_F 
\end{align*}
where $\mathcal{Q}_{[1]} \in \mathbb{R}^{r_1\times \prod_{n\neq 1}r_n}$ is the flattening of core tensor $\mathcal{Q}_{[1]}$ in the first dimension. 
\begin{align*}
    \frac{1}{\sqrt{N}}
    &\left\| \widetilde{A}_{1}^{(1)} \mathcal{Q}_{[1]} \left[\bigotimes_{n \neq 1} \widetilde{A}_{1}^{(n)}\right]'\right\|_F 
    \leq 
    \frac{1}{\sqrt{N}}
    \left\| \widetilde{A}_{1}^{(1)} \right\|_2
    \left\|\mathcal{Q}_{[1]} \left[\bigotimes_{n \neq 1} \widetilde{A}_{1}^{(n)}\right]'\right\|_F
    \\
    &\leq 
    \frac{1}{\sqrt{N}}
    \left\|\widetilde{A}_{1}^{(1)}\right\|_2
    \left\|\mathcal{Q}_{[1]}\right\|_F 
    \left\|\bigotimes_{n \neq 1} \widetilde{A}_{1}^{(n)}\right\|_2= 
    \frac{1}{\sqrt{N}}
    \left\|\mathcal{Q}\right\|_F
    \prod_{n = 1}^d
    \left\|\widetilde{A}_{1}^{(n)}\right\|_2 
\end{align*}
The term $\left\|\mathcal{Q}\right\|_F = N^{-1/2}\left\|\mathcal{A}\right\|_F = O_p(1)$, since each $N_n^{-1/2}U^{(n)}$ has orthonormal columns.
Since weights are formed over estimated $U^{(n)}$, this must be accounted for. 
For each column $\ell_n$ of $\widetilde{A}_{1}^{(n)}$:
\begin{align*}
    \frac{1}{\sqrt{N_n}} (\mathbb{I} - S_{n,\ell_n})U^{(n)}_{\ell_n}   
    &=
    \frac{1}{\sqrt{N_n}} (\mathbb{I} - S_{n,\ell_n})[\hat U^{(n)}_{\ell_n}  + U^{(n)}_{\ell_n}  - \hat U^{(n)}_{\ell_n} ]
    \\
    &= 
    \frac{1}{\sqrt{N_n}} (\mathbb{I} - S_{n,\ell_n})\hat U^{(n)}_{\ell_n} 
    +
    \frac{1}{\sqrt{N_n}} (\mathbb{I} - S_{n,\ell_n})[U^{(n)}_{\ell_n}  - \hat U^{(n)}_{\ell_n} ]
\end{align*}
The first term is a vector bounded by $O_p(h_n)\iota_{N_n}/\sqrt{N_n}$ from standard nonparametric arguments, see \cite{opsomer1999root}.
In the second term, $\max_{\ell_n}\|\mathbb{I} - S_{n,\ell_n}\| = O_p(1)$ by Assumption~\ref{ass:regCondKer}.(iv), as in the proof of Proposition~\ref{lemma:consKernel}.
Hence, bounding the spectral norm of $\widetilde{A}_{1}^{(n)}$ by its Frobenius norm, column by column,
\begin{align*}
    \frac{1}{\sqrt{N_n}}\left\|\widetilde{A}_{1}^{(n)}\right\|_2
    &\leq
    \frac{1}{\sqrt{N_n}}\left\|\widetilde{A}_{1}^{(n)}\right\|_F
    \leq
    \frac{O_p(h_n)}{\sqrt{N_n}}\|\mathbbm{1}_{N_n\times r_n}\|_F
    +
    O_p(1) \frac{1}{\sqrt{N_n}}\left\|U^{(n)}  - \hat U^{(n)} \right\|_F
    \\
    &=
    O_p(\sqrt{r_n}\,h_n + C^{-1}_n).
\end{align*}

Each column $\ell_n$ of $\widetilde{A}_{2,\ell_n}^{(n)}$, $\sum_{\ell_n'\neq \ell_n} S_{n, \ell_n'} U_{\ell_n}^{(n)} =\iota_{N_n} \sum_{\ell_n'\neq \ell_n}O_p(N_nh_n)^{-1/2}$ since $U_{\ell_n}^{(n)}$ are mean zero. 
Hence, each column $\widetilde{A}_{2,\ell_n}^{(n)} = \iota_{N_n}(r_n-1)O_p(N_nh_n)^{-1/2}$. 
Let $\widetilde{\boldsymbol{A}}_2 := {\mathcal{Q}}\times (\widetilde{A}_{2}^{(1)}, \dots, \widetilde{A}_{2}^{(d)} )$, 
\begin{align*}
    \frac{1}{N} \left\langle \widetilde{\boldsymbol{X}}, \widetilde{\boldsymbol{A}}_2 \right\rangle_F
&=
    \frac{1}{N} \left\langle \widetilde{\boldsymbol{X}}, {\mathcal{Q}}\bigtimes_{n=1}^d \widetilde{A}_{2}^{(n)}\right\rangle_F
    =
    \frac{1}{N} \left\langle \widetilde{\boldsymbol{X}}, {\mathcal{Q}}\bigtimes_{n=1}^d (r_n-1) \iota_{N_n} O_p\left(\frac{1}{\sqrt{N_n h_n}}\right)\right\rangle_F
\end{align*}
This is then $o_p(\prod_n N_n^{-1/2})$ since $\widetilde{\boldsymbol{X}}$ is also mean zero:
\begin{align*}
    \frac{1}{N^{1/2}}\prod_{n}(r_n-1)h_n^{-1/2} 
    \frac{1}{N}\left\langle \widetilde{\boldsymbol{X}}, {\mathcal{Q}}\bigtimes_{n=1}^d  \iota_{N_n} \right\rangle_F
    = 
    \frac{1}{N^{1/2}}\prod_{n}(r_n-1)h_n^{-1/2} O_p(N^{-1/2}) \sum_{m_1,\dots,m_d}\big|\mathcal{Q}_{m_1,\dots,m_d}\big|
\end{align*}
where $\sum_{m_1,\dots,m_d}|\mathcal{Q}_{m_1,\dots,m_d}| \leq \prod_n\sqrt{r_n}\,\|\mathcal{Q}\|_F$ by the Cauchy-Schwarz inequality. Hence the term is bounded $O_p\big(N^{-1}\prod_{n}\sqrt{r_n}(r_n-1)h_n^{-1/2}\big)$.
Each $r_n <\infty$, hence this is $o_p(\prod_n N_n^{-1/2})$ for $h_n = c_n\cdot N_n^{-\rho_n}$ and $\rho_n\in (0,1)$.

Now to account for the remaining terms, i.e. the inner product between $\widetilde{\boldsymbol{X}}$ and combinations of $\widetilde{A}_{1}^{(n)}$ and $\widetilde{A}_{2}^{(n)}$. 
Order of dimensions is irrelevant, hence consider $m\in(1,\dots,d-1)$ in, 
\begin{align*}
    \frac{1}{N} \left\langle \widetilde{\boldsymbol{X}}, \mathcal{Q} \times\Big( \widetilde{A}_{2}^{(1)},  \dots, \widetilde{A}_{2}^{(m)}, \widetilde{A}_{1}^{(m+1)}, \dots \widetilde{A}_{1}^{(d)}\Big)\right\rangle_F
\end{align*}
From above $\sum_{\ell_n'\neq \ell_n} S_{n, \ell_n'} U^{(n)} = (r_n - 1)O_p(N_nh_n)^{-1/2}\cdot \iota_{r_n}$. 
Hence, this is, 
\begin{align*}
    \prod_{n=1}^m \frac{r_n-1}{\sqrt{N_n h_n}}\cdot \frac{1}{N} 
    &\left\langle \widetilde{\boldsymbol{X}}, 
    \mathcal{Q} \times\Big( \iota_{N_1}\iota_{r_1}',  \dots, \iota_{N_m}\iota_{r_m}' , \widetilde{A}_{1}^{(m+1)}, \dots \widetilde{A}_{1}^{(d)}\Big)\right\rangle_F
    \\
    &= 
    \left(\prod_{n=1}^m \frac{r_n-1}{\sqrt{N_n h_n}}\right)
    \cdot \frac{1}{N} 
    Tr\left\{\widetilde{X}_{[1:m]}'  
    \bigotimes_{n = 1}^m \iota_{N_n} \Big(\bigotimes_{n = 1}^m \iota_{r_n}\Big)'
    \mathcal{Q}_{[1:m]} \bigotimes_{n = m+1}^d \widetilde{A}_{1}^{(n)}\right\}
    \\
    &=
    \left(\prod_{n=1}^m \frac{r_n-1}{N_n \sqrt{h_n}}\right)
    \cdot 
   \prod_{n'=m+1}^d \frac{1}{N_{n'}} 
    Tr\left\{ 
    \bigotimes_{n = m+1}^d \iota_{N_n} \Big(\bigotimes_{n = 1}^m \iota_{r_n}\Big)'
    \mathcal{Q}_{[1:m]} \bigotimes_{n = m+1}^d \widetilde{A}_{1}^{(n)}\right\}
\end{align*}
The last line uses mean zero $\widetilde{\boldsymbol{X}}$. The rest is rudimentary. 
Lastly, write the last line as, 
\begin{align*}
    \left(\prod_{n=1}^m \frac{r_n-1}{N_n \sqrt{h_n}}\right)
    \cdot 
    \prod_{n'=m+1}^d \frac{1}{N_{n'}} 
    \left\langle
     \bigotimes_{n = 1}^m \iota_{r_n}
    \Big(\bigotimes_{n = m+1}^d \iota_{N_n}\Big)'
    ,
    \mathcal{Q} \times\Big( \mathbb{I}_{r_1},  \dots, \mathbb{I}_{r_m} , \widetilde{A}_{1}^{(m+1)}, \dots \widetilde{A}_{1}^{(d)}\Big)\right\rangle_F
\end{align*}
By Cauchy-Schwarz this is absolutely bounded, 
\begin{align*}
    \left(\prod_{n=1}^m \frac{r_n-1}{N_n \sqrt{h_n}}\right)\prod_{n'=m+1}^d \frac{1}{N_{n'}}
    \Big\|\bigotimes_{n = 1}^m \iota_{r_n}
    \Big(\bigotimes_{n = m+1}^d \iota_{N_n}\Big)'\Big\|_F
    \|\mathcal{Q} \|_F
    \prod_{n=m+1}^d\big\|\widetilde{A}_{1}^{(n)}\big\|_2
    \\
    =
    \left(\prod_{n=1}^m \frac{\sqrt{r_n}(r_n-1)}{N_n \sqrt{h_n}}\right)
    O_p(1)
    \prod_{n'=m+1}^d \frac{1}{\sqrt{N_{n'}}}\big\|\widetilde{A}_{1}^{(n')}\big\|_2
\end{align*}
Hence this is $O_p\big(\prod_{n=1}^m \sqrt{r_n}(r_n-1)(N_n\sqrt{h_n})^{-1}\prod_{n'=m+1}^d (\sqrt{r_{n'}}\,h_{n'} + C_{n'}^{-1})\big)$.

Hence all terms are bounded and
\begin{align*}
    \frac{1}{N} \left|\left\langle \widetilde{\boldsymbol{X}}, \widetilde{\boldsymbol{A}} \right\rangle_F\right|
    =
    \prod_{n=1}^d O_p\Big(C_n^{-1} + \sqrt{r_n}\,h_n + \sqrt{r_n}(r_n-1) \big(N_n \sqrt{h_n}\big)^{-1} \Big)
\end{align*}

\end{proof}

Justified here is a notion to ignore the up-to-rotations consistency of fixed-effects estimates. 
Take invertible matrix $H_n$ from \cite{Bai2009}, such that consistency of fixed-effect proxies are, 
\begin{align*}
    \frac{1}{N_{n}}\left\|\widehat{U}^{(n)} -  {U}^{(n)}H_n\right\|^2 
    =
    \frac{1}{N_{n}}\sum_{i_n = 1}^{N_n}\left\|\widehat{U}_{i_n}^{(n)} -  H_n'{U}_{i_n}^{(n)}\right\|^2
    = O_p(C_{n}^{-2}).
\end{align*}
The HOSVD, 
\begin{align*}
     \mathcal{A} &= \mathcal{Q}\times (U^{(1)}, \dots, U^{(d)})
     = \mathcal{Q}\times (U^{(1)}H_1H_1^{-1}, \dots, U^{(d)}H_dH_d^{-1})
     \\
     &=
     (\mathcal{Q}\times (H_1^{-1}, \dots, H_d^{-1}))\times (U^{(1)}H_1, \dots, U^{(d)}H_d)
     =
     \widetilde{\mathcal{Q}}\times (U^{(1)}H_1, \dots, U^{(d)}H_d)
\end{align*}
Consistency rate can then be analysed v\'is-a-v\'is the decomposition in the last line, and $\|\widetilde{\mathcal{Q}}\|_F \leq \|\mathcal{Q}\|_F\prod_n\|H_n^{-1}\|_2 = O_p(1)$ for invertible $H_n$ with bounded inverse, which leaves the result unaffected.
Mean square consistency over the vector $U_{i_n}^{(n)}$ implies mean square consistency for each component from $m_n=1,\dots,r_n$, 
\begin{align*}
    \frac{1}{N_{n}}\sum_{i_n = 1}^{N_n}\big(\widehat{U}_{i_nm_n}^{(n)} -  H_{n,\cdot m_n}'{U}_{i_n}^{(n)}\big)^2 = O_p(C_{n}^{-2}).
\end{align*}
Hence, can analyse mean square error $\big\{ \frac{1}{N_{n}}\sum_{i_n = 1}^{N_n}\big(\widehat{U}_{i_nm_n}^{(n)}  -H_{n,\cdot m_n}'{U}_{i_n}^{(n)}\big)^2\big\}_{m_n=1}^{r_n}$ without loss. 

\subsection{Linear Kernel}\label{sect:LinearKernel}

Considered here is the special case using the linear kernel: $k(u,v) = u'v$ to generate weights in the weighted-within transformation. 
Under this kernel the proof simplifies significantly. 

\begin{theorem}
    Make Theorem~\ref{lemma:consIterKernel} assumptions. 
    For weights, use the linear kernel with estimates $\widehat{U}^{(n)}$ normalised as in the HOSVD, so that $N_n^{-1/2}\widehat{U}^{(n)}$ is unitary, and $W_n = N_n^{-1}\widehat{U}^{(n)}\widehat{U}^{(n)\,'}$. 
    Then, 
    \begin{align*}
        \norm{\widehat{\beta}_{W} - \beta^0} 
        = 
        \prod_{n=1}^d O_p\Big( C_n^{-1} \Big)
        + 
        O_p\left(\prod_{n = 1}^d\frac{1}{\sqrt{N_{n}}} \right).
    \end{align*}
\end{theorem}

\begin{proof}
    Take, 
\begin{align*}
    \frac{1}{N} \left|\left\langle \widetilde{\boldsymbol{X}}, \widetilde{\boldsymbol{A}} \right\rangle_F\right|
    \leq 
    \frac{1}{N} \left\langle \widetilde{\boldsymbol{X}}, \widetilde{\boldsymbol{X}} \right\rangle_F^{1/2}
    \left\langle \widetilde{\boldsymbol{A}}, \widetilde{\boldsymbol{A}} \right\rangle_F^{1/2}
    = 
    O_p(1) 
    \frac{1}{\sqrt{N}}
    \left\| \widetilde{\boldsymbol{A}}\right\|_F,
\end{align*}

The HOSVD decomposition of $\widetilde{\boldsymbol{A}}$ gives,
\begin{align*}
    \widetilde{\boldsymbol{A}} = \mathcal{Q}\bigtimes_{n =1 }^d \left(\mathbb{I}_{N_n} - N_n^{-1}\widehat{U}^{(n)}\widehat{U}^{(n)\,'}\right)U^{(n)}
\end{align*}
As in the proof of Theorem~\ref{lemma:consIterKernel}, 
\begin{align*}
   \frac{1}{\sqrt{N}}
    \left\| \widetilde{\boldsymbol{A}}\right\|_F 
    \leq 
    \frac{1}{\sqrt{N}}
    \left\|\mathcal{Q}\right\|_F
    \prod_{n = 1}^d
    \left\|\widetilde{A}^{(n)}\right\|_2 
    = 
    O_p(1)\cdot
    \frac{1}{\sqrt{N}}
    \prod_{n = 1}^d
    \left\|\widetilde{A}^{(n)}\right\|_2. 
\end{align*}

Each term $\left\|\widetilde{A}^{(n)}\right\|_2$ can be handled separately. 
\begin{align*}
    \frac{1}{\sqrt{N_n}}&\left\|\widetilde{A}^{(n)}\right\|_2
    = 
    \frac{1}{\sqrt{N_n}}\left\|\left(\mathbb{I}_{N_n} - N_n^{-1}\widehat{U}^{(n)}\widehat{U}^{(n)\,'}\right)U^{(n)}\right\|_2
    \\
    &=
    \frac{1}{\sqrt{N_n}}\left\|\left(\mathbb{I}_{N_n} - N_n^{-1}\widehat{U}^{(n)}\widehat{U}^{(n)\,'}\right)(\widehat{U}^{(n)} +  U^{(n)} - \widehat{U}^{(n)})\right\|_2
    \\
    &\leq 
    \frac{1}{\sqrt{N_n}}\left\|\left(\mathbb{I}_{N_n} - N_n^{-1}\widehat{U}^{(n)}\widehat{U}^{(n)\,'}\right)\widehat{U}^{(n)}\right\|_2
    +  
    \left\|\mathbb{I}_{N_n} - N_n^{-1}\widehat{U}^{(n)}\widehat{U}^{(n)\,'}\right\|_2
    \frac{1}{\sqrt{N_n}}\left\|U^{(n)} - \widehat{U}^{(n)}\right\|_2
\end{align*}
The first term is 0, because $N_n^{-1/2}\widehat{U}^{(n)}$ is unitary, so that $\widehat{U}^{(n)} - N_n^{-1}\widehat{U}^{(n)}\widehat{U}^{(n)\,'}\widehat{U}^{(n)} = \widehat{U}^{(n)} - \widehat{U}^{(n)}$. 
The second term is bounded as follows, since $\mathbb{I}_{N_n} - N_n^{-1}\widehat{U}^{(n)}\widehat{U}^{(n)\,'}$ is an orthogonal projection, with spectral norm at most one,
\begin{align*}
    \left\|\mathbb{I}_{N_n} - N_n^{-1}\widehat{U}^{(n)}\widehat{U}^{(n)\,'}\right\|_2
    \frac{1}{\sqrt{N_n}}\left\|U^{(n)} - \widehat{U}^{(n)}\right\|_2
    &\leq
    \frac{1}{\sqrt{N_n}}\left\|U^{(n)} - \widehat{U}^{(n)}\right\|_F
    =
    O_p(C^{-1}_n)
\end{align*}

Multiplying over all $n =1,\dots, d$ completes the proof. 

\end{proof}

\setlength{\bibsep}{2pt} 
\bibliographystyle{chicago3}
\bibliography{refs}

\end{document}